\documentclass[journal,preprint,comsoc]{IEEEtran}
\usepackage[T1]{fontenc}% optional T1 font encoding
\usepackage[noadjust]{cite} 
\usepackage{amsmath}
\usepackage{amsfonts}
\usepackage{subfigure}
\usepackage{amssymb}
\usepackage{graphicx}% Include figure files
\usepackage{dcolumn}% Align table columns on decimal point
\usepackage{bm}% bold math
\usepackage{braket}
\usepackage{color}
\usepackage{verbatim}
\usepackage{enumitem}
\usepackage{hyperref}
\hypersetup{
	colorlinks=true,
	linkcolor=blue,
	filecolor=blue,
	citecolor=blue,  
	urlcolor=black,
}

\newtheorem{theorem}{Theorem}

\newtheorem{conjecture}[theorem]{Conjecture}

\newtheorem{definition}[theorem]{Definition}

\newtheorem{lemma}[theorem]{Lemma}

\newenvironment{proof}[1][Proof]{\noindent\textbf{#1.} }{\ \rule{0.5em}{0.5em}}
\newenvironment{prooflemma2}[1][Proof of lemma 2]{\noindent\textbf{#1.} }{\ \rule{0.5em}{0.5em}}

\def\be{\begin{equation}}
	\def\ee{\end{equation}}
\def\ba{\begin{eqnarray}}	
	\def\ea{\end{eqnarray}}
\newcommand{\calC}{{\cal C}}
\newcommand{\calD}{{\cal D}}

\newcommand{\calF}{{\cal F}}
\newcommand{\calI}{{\cal I}}
\newcommand{\calN}{{\cal N}}
\newcommand{\calT}{{\cal T}}
\newcommand{\rmconv}{{\rm Conv}}
\newcommand{\bbr}[1]{\left(#1\right)}
\newcommand{\1}{^{(1)}}
\usepackage[cmintegrals]{newtxmath}
\begin{document}
%\preprint{MIT-CTP/4917}
\title{Superadditivity in trade-off capacities of quantum channels}
%
%
% author names and IEEE memberships
% note positions of commas and nonbreaking spaces ( ~ ) LaTeX will not break
% a structure at a ~ so this keeps an author's name from being broken across
% two lines.
% use \thanks{} to gain access to the first footnote area
% a separate \thanks must be used for each paragraph as LaTeX2e's \thanks
% was not built to handle multiple paragraphs
%

\author{Elton~Yechao~Zhu,~Quntao~Zhuang,~Min-Hsiu~Hsieh,~\IEEEmembership{Senior Member,~IEEE},~and~Peter~W.~Shor
\thanks{E.Y. Zhu is with the Center of Theoretical Physics and Department of Physics, Massachusetts Institute of Technology, Cambridge, MA, 02139 USA (e-mail: eltonzhu@mit.edu).}% <-this % stops a space
\thanks{Q. Zhuang is with the Research Laboratory of Electronics and Department of Physics, Massachusetts Institute of Technology, Cambridge, MA, 02139 USA (e-mail: quntao@mit.edu).}% <-this % stops a space
\thanks{Min-Hsiu Hsieh is with the Centre of Quantum Software and Information (UTS$\ket{\text{QSI}}$), University of Technology Sydney, Australia (e-mail: Min-Hsiu.Hsieh@uts.edu.au).}
\thanks{P.W. Shor is with the Department of Mathematics, Massachusetts Institute of Technology, Cambridge, MA, 02139 USA (e-mail: shor@math.mit.edu).}}% <-this % stops a space% <-
\maketitle

\begin{abstract}
In this article, we investigate the additivity phenomenon in the dynamic capacity of a quantum channel for trading classical communication, quantum communication and entanglement. Understanding such additivity property is important if we want to optimally use a quantum channel for general communication purpose. However, in a lot of cases, the channel one will be using only has an additive single or double resource capacity, and it is largely unknown if this could lead to an superadditive double or triple resource capacity. For example, if a channel has an additive classical and quantum capacity, can the classical-quantum capacity be superadditive? In this work, we answer such questions affirmatively. 

We give proof-of-principle requirements for these channels to exist. In most cases, we can provide an explicit construction of these quantum channels. The existence of these superadditive phenomena is surprising in contrast to the result that the additivity of both classical-entanglement and classical-quantum capacity regions imply the additivity of the triple capacity region. 
\end{abstract}

% Note that keywords are not normally used for peerreview papers.

% For peer review papers, you can put extra information on the cover
% page as needed:
% \ifCLASSOPTIONpeerreview
% \begin{center} \bfseries EDICS Category: 3-BBND \end{center}
% \fi
%
% For peerreview papers, this IEEEtran command inserts a page break and
% creates the second title. It will be ignored for other modes.
\IEEEpeerreviewmaketitle

\section{Introduction}
\IEEEPARstart{I}{n} studying classical communication, Shannon developed powerful probabilistic tools that connect the theoretic throughput of a channel to an entropic quantity defined on a single use of the channel \cite{Shannon48}. Shannon's noiseless channel coding theorem involves a random coding strategy to prove achievability and entropic inequalities that show optimality, \textit{i.e.,} the converse. This methodology has now become standard in proving finite or asymptotic optimal resource conversions in information theory. 
%Moreover, as a consequence of the single-lettered capacity expression, one important implication of this seminal result is that optimal coding for bipartite communication can be done over a single channel use.

Quantum Shannon information starts by mimicking classical information theory: typical sets can be generalized to typical subspaces to prove achievability while various entropic inequalities, such as the quantum data processing inequality, can be used to prove the converse. However, the differences between quantum and classical Shannon information are also significant. On one hand, additional resources available in the quantum domain diversify the allowable capacities, resulting in trade-off regions for the resources that are consumed or generated \cite{Devetak04, Devetak08, HsiehWilde10}. The most common, and useful, quantum resource in communication settings is quantum entanglement. Unlike classical shared randomness, which does not increase a classical channel's capability to send more messages, preshared quantum entanglement will generally increase the throughput of a quantum channel for sending classical messages or quantum messages or both \cite{Bennett92, Bennett99, Bennett02, Devetak04, Hsieh10, Datta13}. It thus makes sense to consider the trade-off capacity regions among these three useful resources: entanglement, classical communication, and quantum communication, and this was done in Ref.~\cite{HsiehWilde10}. The result in Ref. \cite{HsiehWilde10} further shows that a coding strategy that exploits the channel coding of these three resources as a whole performs better than strategies that do not take advantage of channel coding. %We remark that these three resource trade-offs are also the main focus of this paper; however, other triple resource trade-offs do exist \cite{Wilde2012}.

On the other hand, single-lettered channel capacity formulas in the classical regime generally become intractable regularized capacity formulas in the quantum regime \cite{Holevo98, Schumacher97, Lloyd97, Shor02, Devetak05}. In other words, evaluation of these capacity quantities requires optimizing channel inputs over an arbitrary finite number of uses of a given channel. This largely blocks our understanding of how quantum channels behave. An extreme example shows the existence of two quantum channels that cannot be used to send a quantum message individually but will have a positive channel capacity when both are used simultaneously \cite{SmithYard08}. However, there are also several examples showing that when additional resources are used to assist, the corresponding assisted capacity will also become additive. The classical capacity over quantum channels is generally superadditive; however, when assisted by a sufficient amount of entanglement, the entanglement-assisted capacity becomes additive \cite{Bennett99, Shor04}. The quantum capacity also exhibits similar properties. When assisted by either entanglement \cite{Devetak04, Devetak08} or an unbounded symmetric side channel \cite{Smith08}, its assisted quantum capacity becomes additive.

This superadditive property of quantum channel capacities has accordingly attracted significant attention. Hastings \cite{Hastings09} proved that the classical capacity over quantum channels is not additive, a result built upon earlier developments by Hayden-Winter \cite{Hayden08} and Shor \cite{Shor04additivity}. Recently, three of us showed a rather perplexing result \cite{ZhuZhuangShor17}: when assisted by an insufficient amount of entanglement, a channel's classical capacity could be superadditive regardless of whether the unassisted classical capacity is additive or not. Further, the additive property of the entanglement-assisted classical capacity shows a form of phase transition. Even if the channel is additive when assisted by a sufficient amount of entanglement or no entanglement at all, it can still be superadditive when assisted with an insufficient amount of entanglement. This phenomenon indicates that quantum channels behave fundamentally differently from classical channels, and our understanding of it is still quite limited.

This paper is inspired by, and aims to extend Ref. \cite{ZhuZhuangShor17}. Will additivity of single or double resource capacities always lead to additivity of a general resource trade-off region? We will study superadditi vity in a general framework that considers the three most common resources of: entanglement, noiseless classical communication and quantum communication. Our results show that (i) additivity of single resource capacities of a quantum channel does not generally imply additivity of double resource capacities, except for the known result \cite{Devetak04} that an additive quantum capacity yields an additive entanglement-assisted quantum capacity region (see Table~\ref{tableDouble}); and (ii) additive double resource capacities does not generally imply an additive triple resource capacity, except for the known case \cite{Hsieh10} that additive classical-entanglement and classical-quantum capacity regions yield an additive triple dynamic capacity (see Table~\ref{tableTriple}). These results again demonstrate how complex a quantum channel can be, and further investigation is required. 
%\textcolor{red}{Not finished}

The paper is structured as follows. Section \ref{Sec:Preliminaries} introduces the various definitions, notations and previous results on the triple resource quantum Shannon theory. Section \ref{Sec:Summary} summarizes the various superadditivity results that we establish in the paper. Section \ref{Sec:Framework} establishes the switch channel that we use for all our constructions, and how this reduces the triple resource trade-off formula. Section \ref{Sec:Construction} gives a detailed construction of all the possible superadditivity phenomena.

\section{Preliminaries}\label{Sec:Preliminaries}
In this section, we give definitions of basic entropic quantities used in the paper. We also describe the dynamic capacity theorem. Special cases of this include the various single and double resource capacities. Finally, we define the elementary channels that will be used in our explicit constructions.

A bipartite quantum state $\sigma_{AB}$ is a positive semi-definite matrix in Hilbert space $\cal{H}_A\otimes\cal{H}_B$ with trace one. We define the von Neumann entropy, conherent information and quantum mutual information of $\sigma_{AB}$, respectively, as follows:
\begin{eqnarray*}
S(AB)_\sigma &=& -\text{Tr }[ \sigma_{AB}\log \sigma_{AB}],\\
I(A\rangle B)_{\sigma} &=&  S(B)_{\sigma} -S(AB)_{\sigma}, \\
I(A;B)_\sigma &=& S(A)_\sigma +I(A\rangle B)_{\sigma}, 
\end{eqnarray*}
where $S(A)_{\sigma}$ is the von Neumann entropy of the reduced state $\sigma_A = \text{Tr}_B [\sigma_{AB}]$.

For an ensemble $\{p(x), \sigma^x_{AB}\}_{x\in\mathcal{X}}$, let
\begin{equation*}
\sigma_{XAB}=\sum_{x\in\mathcal{X}} p(x) |x\rangle \langle x|_X\otimes\sigma^x_{AB},
\end{equation*}
where $\{|x\rangle\}$ forms a fixed orthonormal (computational) basis in Hilbert Space $\mathcal{H}_X$. We need the following information quantities as well: 
\begin{eqnarray}
I(A\rangle BX)_{\sigma} &=& \sum_x p(x) I(A\rangle B)_{\sigma_x},\label{eq_conConherent} \\
I(A;B|X)_\sigma &=& \sum_x p(x) I (A;B)_{\sigma_x}, \label{eq_conMutual}\\
I(AX;B)_\sigma &=& I(X;B)_\sigma + I(A;B|X)_\sigma,  \label{eq_Holevo}
\end{eqnarray}
where $I(A\rangle BX)_{\sigma}$ and $I(A;B|X) $ in Eqs.~(\ref{eq_conConherent}) and (\ref{eq_conMutual}) are the conditional coherent information and the conditional mutual information, respectively. $I(X;B)_\sigma$ in Eq.~(\ref{eq_Holevo}) is the Holevo information of $\sigma_{XB} = \text{Tr}_A [\sigma_{XAB}]$.

A quantum channel $\mathcal{N}$ is a completely positive and trace-preserving map. With it, we can transmit either classical or quantum information or both with possible entanglement assistance between the sender and the receiver \cite{Hsieh10}. More generally, the authors in Ref. \cite{HsiehWilde10} proved the following capacity theorem that involves a noisy quantum channel $\mathcal{N}$ and the three resources mentioned above; namely, classical communication (C), quantum communication (Q) and quantum entanglement (E). 
\begin{theorem}[CQE trade-off \cite{HsiehWilde10}] 
\label{thm_CQE}
The dynamic capacity region $ \calC_{CQE}\left(\calN\right) $ of a quantum channel $ \calN $ is equal to the following expression:
\be
\calC_{CQE}\left(\calN\right)=\overline{\bigcup_{k=1}^{\infty}\frac{1}{k}\calC^{(1)}_{CQE}\left(\calN^{\otimes k}\right)},\nonumber
\ee
where the overbar indicates the closure of a set. The region $ \calC^{(1)}_{CQE}\left(\calN\right) $ is equal to the union of the state-dependent regions $ \calC_{CQE,\sigma}^{(1)}\left(\calN\right) $:
\be
\calC_{CQE}^{(1)}\left(\calN\right)\equiv \bigcup_{\sigma}\calC_{CQE,\sigma}^{(1)}\left(\calN\right).\nonumber
\ee
The state-dependent region $ \calC_{CQE,\sigma}^{(1)}\left(\calN\right) $ is the set of all rates $ C $, $ Q $ and $ E $, such that
\begin{align}
C+2Q&\leq I(AX;B)_\sigma, \label{eq:CQ}\\
Q+E&\leq I(A\rangle BX)_\sigma,\label{eq:QE}\\
C+Q+E&\leq I(X;B)_\sigma+I(A\rangle BX)_\sigma. \label{eq:CQE}
\end{align}
The above entropic quantities are with respect to a classical-quantum state (cq state) $ \sigma_{XAB} $, where
\be\label{eq:cqinput}
\sigma_{XAB}\equiv\sum_xp(x)\ket{x}\bra{x}_X\otimes\calN_{A'\to B}\left(\phi^x_{AA'}\right),
\ee
and the states $ \phi^x_{AA'} $ are pure. 
\end{theorem}

We say that the dynamic capacity of a channel $\calN$ is \emph{additive} if 
\be \label{eq_addtiviedynamicC}
\calC_{CQE}\left(\calN\right)=\calC_{CQE}^{(1)}\left(\calN\right).
\ee

The dynamic capacity region $\calC_{CQE}(\calN)$ in Theorem~\ref{thm_CQE} allows us to recover known capacity theorems by choosing certain $(C,Q,E)$ in Eqs.~(\ref{eq:CQ})-(\ref{eq:CQE}) as follows:
\begin{itemize}
\item the classical capacity $\calC_{C}(\calN)$ when choosing $Q=E=0$ \cite{Holevo98,Schumacher97};
\item the quantum capacity $\calC_{Q}(\calN)$ when choosing $C=E=0$ \cite{Lloyd97,Shor02,Devetak05};
\item the classical and quantum capacity $\calC_{CQ}(\calN)$ when choosing $E=0$ (CQ trade-off) \cite{DevetakShor05};
\item the entanglement assisted classical capacity $\calC_{CE}(\calN)$ when choosing $Q=0$ (CE trade-off) \cite{Bennett02,Shor04};
\item the entanglement assisted quantum capacity $\calC_{QE}(\calN)$ when choosing $C=0$ (QE trade-off) \cite{Devetak04,Devetak08};
\end{itemize}
Additivity of these special cases follows similarly from Eq.~(\ref{eq_addtiviedynamicC}).

We note that the dynamic capacity region is concave, as a convex combination of any two points in the region can be achieved by a time-sharing strategy, \textit{i.e.,} using the channel for a fraction of uses to achieve one point, and using it for the other fraction to achieve the second point.

Below we will briefly describe a few channels which we will repeatedly use. 
\begin{definition}
A \textit{Hadamard channel} is a quantum channel whose complementary channel is entanglement breaking. Suppose $ \Psi_{A'\to B} $ is a Hadamard channel, with the complementary channel $ \Psi^c_{A'\to E} $. Then there is a degrading map $ \calD_{B\to E} $ such that
\be
\Psi^c_{A'\to E}=\calD_{B\to E}\circ\Psi_{A'\to B}.\nonumber
\ee
Moreover, $ \calD $ can be decomposed as
\be
\calD_{B\to E}=\calD^2_{Y\to E}\circ\calD^1_{B\to Y},\nonumber
\ee
where $ Y $ is a classical variable.	
\end{definition}
A Hadamard channel has an additive quantum dynamic capacity region, when tensored with an arbitrary quantum channel \cite{Bradler10}. Examples of Hadamard channels include the qubit dephasing channel, $ 1\to N $ cloning channels, and the Unruh channel. We'll define the qubit dephasing channel below, but refrain from giving definitions of other Hadamard channels, since their exact forms are not needed for understanding this work. We refer the interested readers to Ref. \cite{Bradler10} for more details and properties of these channels.
\begin{definition}
The \textit{qubit dephasing channel} $ \Psi^{\textrm{dph}}_\eta $, with dephasing probability $\eta$, is defined as
\be
\Psi^{\textrm{dph}}_\eta\left(\rho\right)=(1-\eta)\rho+\eta Z\rho Z.\nonumber
\ee	
\end{definition}

\begin{definition}
The \textit{qubit depolarizing channel} $ \Psi^{\textrm{dpo}}_p $, with depolarizing probability $p$, is defined as
\be
\Psi^{\textrm{dpo}}_p\left(\rho\right)=(1-p)\rho+p\frac{I}{2}.\nonumber
\ee	
\end{definition}
The qubit depolarizing channel is known to have an additive classical capacity \cite{King02}, but a superadditive quantum capacity \cite{Divincenzo98}.
%There are two special cases. When $ p=0 $, we get the noiseless qubit channel $ \calI $. When $ p=1 $, we get the completely qubit depolarizing channel.

\begin{definition}
A \textit{random orthogonal channel} $ \Psi^{\textrm{ro}} $ is defined as
\be
\Psi^{\textrm{ro}}\left(\rho\right)=\sum_{i=1}^DP_iO_i\rho O_i^\intercal,\nonumber
\ee
where $O_i$ are chosen from the orthogonal group and the probabilities $P_i$ are roughly equal.	
\end{definition}
For $ 1\ll D\ll N $, with $ N $ the input dimension, such a channel will have a subadditive minimum output entropy with high probability \cite{Hastings09}.
\begin{definition}
Consider an arbitrary channel $ \Psi_{C\to B} $. Append a register $ R $ to the input, with a set of orthonormal bases $\{\ket{j}\}$ and $ |R|=|B|^2 $. We define its \textit{unitally extended channel} \cite{Shor04additivity, Fukuda07} $ \Phi_{RC\to B} $ as
\be
\Phi_{RC\to B}(\rho_{RC})=\sum_jX(j)\Psi_{C\to B}\left(\bra{j}\rho_{RC}\ket{j}_R\right)X(j)^\dagger,\label{eq:unitaryencoding}
\ee
where $ \{X(j):j\in \{1,\dots,|R|\}\} $ are the Heisenberg-Weyl operators.	
\end{definition}
The unital extension of a random orthogonal channel will have a superadditive classical capacity with high probability \cite{Shor04additivity}.

\section{Summary of Results}\label{Sec:Summary}

We summarize all of our results here. We will denote the single capacity region by a single letter, e.g. C for $\calC_{C}(\calN)$. We will also use short notation for double and triple trade-off regions, e.g. CE for $\calC_{CE}(\calN)$ and CQE for $\calC_{CQE}(\calN)$. We will use the arrow notation, with ``$ \to $'' meaning additivity of the left-hand side capacity implies additivity of the right-hand side capacity, and ``$ \not{\to} $'' meaning additivity of the left-hand side capacity does not imply additivity of the right-hand side capacity.
\begin{table}
\centering
\begin{tabular}{ |c||c|c|c|  }
 \hline
& \multicolumn{3}{c|}{Imply the additivity of} \\
 \hline
Additive capacities  &CE&CQ&QE\\
 \hline
C   & N~\cite{ZhuZhuangShor17} & N~\cite{Divincenzo98} & N~\cite{Divincenzo98}\\
Q   & N~(Sec\ref{Subsec:QaddCQnonadd}) & N~(Sec\ref{Subsec:QaddCQnonadd}) & Y~\cite{Devetak08}\\
C,Q$\Leftrightarrow$ C,QE & N~(Sec\ref{Subsec:CnQaddCEnonadd}) & N~(Sec\ref{Subsec:CnQaddCQnonadd}) & Y\cite{Devetak08}\\
 \hline
\end{tabular}
\caption{Summary of results for double resources.``N'' stands for ``does not imply additivity'', while ``Y'' means ``implies additivity''.
\label{tableDouble}
}
\end{table}
\begin{table}
\centering
\begin{tabular}{ |c||c|  }
 \hline
& Imply the additivity of \\
 \hline
Additive capacities  &CQE\\
 \hline
QE   & N (Sec\ref{Subsec:QaddCQnonadd})\\
CQ   & N (Sec\ref{Subsec:CQadd})\\
CE   & N (Sec\ref{Subsec:CEadd}) \\
CE,Q$\Leftrightarrow$CE,QE & N~(Sec\ref{Subsec:CEnQadd}) \\
CE,CQ & Y\cite{Hsieh10}\\
 \hline
\end{tabular}
\caption{Summary of results for triple resources. ``N'' stands for ``does not imply additivity'', while ``Y'' means ``implies additivity''.
\label{tableTriple}
}
\end{table}

\begin{comment}
\subsection{Single resource}
\begin{enumerate}
	\item $ C\not{\to}Q $: There exists a quantum channel $ \calN $, such that its classical capacity is additive, but its quantum capacity is superadditive.
%	\textit{i.e.} $ \exists $ a quantum channel $ \calN $ \textit{s.t.}
%	\be
%	\calC_C\bbr{\calN}=\calC_C\1\bbr{\calN}\nonumber
%	\ee
%	but
%	\be
%	\calC_Q\bbr{\calN}>\calC_Q\1\bbr{\calN}.\nonumber
%	\ee
	The depolarizing channel is an example. So we will not explain further.
	
	\item $ Q\not{\to}C $: There exists a quantum channel $ \calN $, such that its quantum capacity is additive, but its classical capacity is superadditive, \textit{i.e.} $ \exists $ a quantum channel $ \calN $ \textit{s.t.}
	\be
	\calC_Q\bbr{\calN}=\calC_Q\1\bbr{\calN}\nonumber
	\ee
	but
	\be
	\calC_C\bbr{\calN}>\calC_C\1\bbr{\calN}.\nonumber
	\ee
\end{enumerate}
\end{comment}
\subsection{Double resources (see table~\ref{tableDouble})}
\begin{enumerate}
	\item CE:
	\begin{enumerate}
		\item $ C\not{\to} CE $ \cite{ZhuZhuangShor17}: There exists a quantum channel $ \calN $, such that its classical capacity is additive, but its CE trade-off capacity region is superadditive.
%		, \textit{i.e.}, $ \exists $ a quantum channel $ \calN $ \textit{s.t.}
%		\be
%		\calC_C\left(\calN\right)=\calC^{(1)}_C\left(\calN\right)\nonumber
%		\ee
%		but
%		\be
%		\calC_{CE}\left(\calN\right)\supsetneq \calC^{(1)}_{CE}\left(\calN\right).\nonumber
%		\ee
		We will give a simplified construction in Sec \ref{Subsec:CaddCEnonadd}.
		\item $ C,Q\not{\to} CE $: There exists a quantum channel $ \calN $, such that its classical and quantum capacities are both additive, but its CE trade-off capacity region is superadditive, \textit{i.e.,}, $ \exists $ a quantum channel $ \calN $ \textit{s.t.}
		\be
		\calC_C\bbr{\calN}=\calC^{(1)}_C\bbr{\calN}\nonumber
		\ee
		and
		\be
		\calC_{Q}\bbr{\calN}=\calC_{Q}\1\bbr{\calN} \nonumber
		\ee
		but
		\be
		\calC_{CE}\left(\calN\right)\supsetneq \calC^{(1)}_{CE}\left(\calN\right).\nonumber
		\ee
		An explicit construction of $\calN$ is given in Sec \ref{Subsec:CnQaddCEnonadd}.
	\end{enumerate}
	\item $ Q\to QE $ \cite{Devetak08}: For all quantum channels $ \calN $, if its quantum capacity is additive, then its QE trade-off capacity region is always additive. %, \textit{i.e.} $ \forall $ quantum channels $ \calN $, if
%	\be
%	\calC_Q\left(\calN\right)=\calC^{(1)}_Q\left(\calN\right),\nonumber
%	\ee
%	then
%	\be
%	\calC_{QE}\left(\calN\right)=\calC^{(1)}_{QE}\left(\calN\right).\nonumber
%	\ee
\pagebreak
	\item CQ:
	\begin{enumerate}
		\item $ C\not{\to}CQ $~\cite{Divincenzo98}: There exists a quantum channel $ \calN $, such that its classical capacity is additive, but its CQ trade-off capacity region is superadditive. %, \textit{i.e.} $ \exists $ a quantum channel $ \calN $ \textit{s.t.}
%		\be
%		\calC_C\left(\calN\right)=\calC^{(1)}_C\left(\calN\right)\nonumber
%		\ee
%		but
%		\be
%		\calC_{CQ}\left(\calN\right)\supsetneq \calC^{(1)}_{CQ}\left(\calN\right).\nonumber
%		\ee
		The depolarizing channel has a superadditive quantum capacity and hence a superadditive CQ trade-off capacity, while its classical capacity is additive.
		\item $ Q\not{\to} CQ $: There exists a quantum channel $ \calN $, such that its quantum capacity is additive, but its CQ trade-off capacity region is superadditive, \textit{i.e.,} $ \exists $ a quantum channel $ \calN $ \textit{s.t.}
		\be
		\calC_Q\left(\calN\right)=\calC^{(1)}_Q\left(\calN\right)\nonumber
		\ee
		but
		\be
		\calC_{CQ}\left(\calN\right)\supsetneq \calC^{(1)}_{CQ}\left(\calN\right).\nonumber
		\ee
		A construction of this example quantum channel is given in Sec~\ref{Subsec:QaddCQnonadd}.
		\item $ C,Q\not{\to}CQ $: Moreover, there exists a quantum channel $ \calN $, such that its classical and quantum capacities are additive, but its CQ trade-off capacity region is superadditive, \textit{i.e.,} $ \exists $ a quantum channel $ \calN $ \textit{s.t.}
		\be
		\calC_C\left(\calN\right)=\calC^{(1)}_C\left(\calN\right)\nonumber
		\ee
		and
		\be
		\calC_Q\left(\calN\right)=\calC^{(1)}_Q\left(\calN\right),\nonumber
		\ee
		but
		\be
		\calC_{CQ}\left(\calN\right)\supsetneq \calC^{(1)}_{CQ}\left(\calN\right).\nonumber
		\ee
		A construction of this example quantum channel is given in Sec~\ref{Subsec:CnQaddCQnonadd}.
	\end{enumerate}
\end{enumerate}
\subsection{Triple resources (see table~\ref{tableTriple})}
\begin{enumerate}
	\item $ CE\not\to CQE $: There exists a quantum channel $ \calN $ such that its CE trade-off capacity region is additive, but its dynamic capacity region is superadditive, \textit{i.e.,} $ \exists $ a quantum channel $ \calN $ \textit{s.t.}
	\be
	\calC_{CE}\left(\calN\right)= \calC^{(1)}_{CE}\left(\calN\right)\nonumber
	\ee
	but
	\be
	\calC_{CQE}\left(\calN\right)\supsetneq \calC^{(1)}_{CQE}\left(\calN\right).\nonumber
	\ee
	An example is constructed in Sec~\ref{Subsec:CEadd}.
	\item $ CE,Q\not\to CQE $: There exists a quantum channel $ \calN $ such that its quantum capacity and its CE trade-off capacity region are additive, but its dynamic capacity region is superadditive, \textit{i.e.,} $ \exists $ a quantum channel $ \calN $ \textit{s.t.}
	\be
	\calC_Q\bbr{\calN}=\calC_Q\1\bbr{\calN}\nonumber
	\ee
	and
	\be
	\calC_{CE}\bbr{\calN}=\calC_{CE}\1\bbr{\calN},\nonumber
	\ee
	but
	\be
	\calC_{CQE}\bbr{\calN}\supsetneq\calC_{CQE}\1\bbr{\calN}.\nonumber
	\ee
	An example is constructed in Sec~\ref{Subsec:CEnQadd}.
	\item $ CQ\not\to CQE $: There exists a quantum channel $ \calN $ such that its CQ trade-off capacity region is additive, but its dynamic capacity region is superadditive, \textit{i.e.,} $ \exists $ a quantum channel $ \calN $ \textit{s.t.}
	\be
	\calC_{CQ}\bbr{\calN}=\calC_{CQ}\1\bbr{\calN}\nonumber
	\ee
	but
	\be
	\calC_{CQE}\bbr{\calN}\supsetneq\calC_{CQE}\1\bbr{\calN}.\nonumber
	\ee
	An example is given in Sec~\ref{Subsec:CQadd}.
	\item $ CE,CQ\to CQE $ \cite{Hsieh10}: If a quantum channel $ \calN $ has additive CE and CQ trade-off capacity regions, then its dynamic capacity region is also additive. %, \textit{i.e.} if
%	\be
%	\calC_{CE}\bbr{\calN}=\calC_{CE}\1\bbr{\calN}\nonumber
%	\ee
%	and
%	\be
%	\calC_{CQ}\bbr{\calN}=\calC_{CQ}\1\bbr{\calN},\nonumber
%	\ee
%	then
%	\be
%	\calC_{CQE}\bbr{\calN}=\calC_{CQE}\1\bbr{\calN}.\nonumber
%	\ee
	This statement is first observed in Ref.~\cite{Hsieh10}, and an explicit argument can be found in Ref.~\cite{Bradler10}.
\end{enumerate}

\section{Framework}\label{Sec:Framework}

This section presents technical tools that we require for demonstration of superadditivity in trade-off capacities. 
We first define the concept of switch channels. 
\begin{definition}
A \textit{switch channel} $ \calN_{MC\to B}$ between $\calN^0_{C\to B}$ and $\calN^1_{C\to B}$ with $ M $ being a 1-bit switch register is defined as
\begin{align*}
&\calN_{MC\to B}\left(\rho_{MC}\right)\\
=&\calN^0_{C\to B}\left(\bra{0}\rho_{MC}\ket{0}_M\right)+\calN^1_{C\to B}\left(\bra{1}\rho_{MC}\ket{1}_M\right).
\end{align*}
\end{definition}

In quantum information theory, switch channels were first used in Ref.~\cite{Bennett02} to demonstrate the existence of quantum channels such that the quantum capacity is nonzero, but for which pre-shared entanglement does not improve the classical capacity. Subsequently, they are used in Ref.~\cite{Elkouss15} to show the superadditivity of private information, with an alternative definition. Recently, they are also used in Ref.~\cite{ZhuZhuangShor17} to show the superadditivity of the classical capacity with limited entanglement assistance.

One immediate difficulty is that, even if $ \calN^0 $ and $ \calN^1 $ are well-studied, the dynamic capacity region of $ \calN $ may not always have a simple expression in terms of those of $ \calN^0 $ and $ \calN^1 $. This is due to the fact that the switch register $ M $ can be in a superposition state. However, if $ \calN^0 $ and $ \calN^1 $ are unitally extended channels, then the dynamic capacity region of $ \calN $ does have a simple expression.

\begin{lemma}\label{lem:switchunital}
Consider a switch channel $ \calN_{A'\to B} $ between $ \calN^0_{RC\to B} $ and $ \calN^1_{RC\to B} $, with input partition $ A'=MRC $ and $ M $ being a switch register. Here $ \calN^0_{RC\to B} $ and $ \calN^1_{RC\to B} $ are unital extensions of $ \Psi^0_{C\to B} $ and $ \Psi^1_{C\to B} $ respectively. Then
\be
\calC^{(1)}_{CQE}\left(\calN\right)=\rmconv\left(\calC_{CQE}^{(1)}\left(\calN^0\right),\calC_{CQE}^{(1)}\left(\calN^1\right)\right),\nonumber
\ee
where Conv denotes the convex hull of points from the two sets.

If the quantum dynamic capacity region for $ \calN^0\otimes\Psi $ is additive for any $ \Psi $, then we also have
\be
\calC_{CQE}\left(\calN\right)=\rmconv\left(\calC_{CQE}\left(\calN^0\right),\calC_{CQE}\left(\calN^1\right)\right).\nonumber
\ee
\end{lemma}

The rest of this section is devoted to the proof of this lemma.

Firstly, we note that switch channels and unitally extended channels fall under a broader class of channels that we call partial classical-quantum channels (partial cq channels). 
\begin{definition}
A channel $ \Psi_{RC\to B} $ is a \textit{partial cq channel} if there exists a noiseless classical channel $ \Pi_{R\to R} $ with orthonormal basis $ \{\ket{j}_R\} $, such that
\be \label{eq:cqchannnel}
\Psi_{RC\to B}=\Psi_{RC\to B}\circ \Pi_{R\to R}.
\ee	
\end{definition}
If there is no register $ C $, then such channels are classical-quantum channels (cq channels).

For partial cq channels, one can always assume inputs are cq states with respect to the input partition $ R $ and $ C $ for the purpose of evaluating capacities, as we show in Lemma \ref{lem:cqCQE} below.
\begin{lemma}\label{lem:cqCQE}
If $ \Psi_{A'\to B} $ is a partial cq channel with partition $ A'=RC $, then the optimal trade-off surface of the 1-shot dynamic capacity region $ \calC_{CQE}^{(1)}\left(\Psi\right) $ can be achieved with respect to cq states $ \sigma_{XAB}=\Psi_{A'\to B}\left(\rho_{XAA'}\right) $, where $ \rho_{XAA'} $ is of the form
\be \label{eq:lem2input}
\rho_{XAA'}=\sum_{x,j}p(x,j)\ket{x,j}\bra{x,j}_X\otimes\ket{j}\bra{j}_R\otimes\phi^{xj}_{AC}. 
\ee
\end{lemma}
\begin{proof}
We will show that, for any input state
\be\label{eq:lem2generalinput}
\varrho_{XAA'}=\sum_xp(x)\ket{x}\bra{x}_X\otimes\phi^x_{AA'},
\ee
with its output state 
$$\varsigma_{XAB}\equiv\Psi_{A'\to B}\left(\varrho_{XAA'}\right)=\sum_xp(x)\ket{x}\bra{x}_X\otimes\varsigma^x_{AB},$$ 
where $\varsigma^x_{AB}= \Psi_{A'\to B}\left(\phi^x_{AA'}\right)$,
there exists a corresponding state $\rho_{XAA'}$, in the form of Eq.~(\ref{eq:lem2input}), which can achieve the same rate, if not better.

In fact, the state $\rho_{XAA'}$ can be obtained by applying $ \Pi_{R\to R} $ on $\varrho_{XAA'}$ and expanding its classical register $X$. This can be achieved by the following quantum instrument $\calT: R\to RX_R$,
\[
\calT\left(\psi_R\right):=\sum_j\bra{j}\psi_R\ket{j}\ket{j}\bra{j}_R\otimes\ket{j}\bra{j}_{X_R}
\]
so that
\begin{eqnarray}
\rho_{XAA'} &=& \calT (\varrho_{XAA'})\nonumber \\
&=& \sum_{x,j}p(x,j)\ket{x,j}\bra{x,j}_X\otimes\ket{j}\bra{j}_R\otimes\phi^{xj}_{AC} \label{eq_qint}
\end{eqnarray}
where we abuse the notation $X$ to denote $XX_R$ in Eq.~(\ref{eq_qint}), $p(x,j)\equiv  p(x)p(j|x)$, and $p(j|x)= \text{Tr}[\ket{j}\bra{j} \phi^{x}_{AC}] $.

Let $ \sigma_{XAB}=\Psi_{A'\to B}\left(\rho_{XAA'}\right) $. Then
\be
\sigma_{XAB}=\sum_{x,j}p(x,j)\ket{x,j}\bra{x,j}_X\otimes\sigma^{xj}_{AB}\nonumber
\ee
where
\be
\sigma^{xj}_{AB}=\Psi_{A'\to B}\left(\ket{j}\bra{j}_R\otimes\phi^{xj}_{AC}\right).\nonumber
\ee

\begin{comment}
the state $\rho_{XAA'}$ is obtained by applying $ \Pi_{R\to R} $ on $ \phi^x_{AA'} $ in Eq.~(\ref{eq:lem2generalinput}): $\Pi_{R\to R}\left(\phi^x_{AA'}\right)=\sum_jp(j|x)\ket{j}\bra{j}_R\otimes\phi^{xj}_{AC},$ resulting in 
\be
\sigma_{XAB}=\sum_{x,j}p(x,j)\ket{x}\bra{x}_X\otimes\sigma^{xj}_{AB}\nonumber
\ee
where
\be
\sigma^{xj}_{AB}=\Psi_{A'\to B}\left(\ket{j}\bra{j}_R\otimes\phi^{xj}_{AC}\right)\nonumber
\ee
and
\be
p(x,j)\equiv  p(x)p(j|x)\label{eq:jprob}.
\ee
\end{comment}
It follows that
\be
\varsigma^{x}_{AB}=\sum_jp(j|x)\sigma^{xj}_{AB}.\label{eq:jdecomp}
\ee

Since the dynamic capacity region is fully determined by the three entropic quantities $ I(AX;B)_\sigma $, $ I(A\rangle BX)_\sigma $ and $ I(X;B)_\sigma $ in Eqs.~(\ref{eq:CQ})-(\ref{eq:CQE}), it suffices to show that all three entropic quantities evaluated on $ \rho_{XAA'} $ are greater than those evaluated on $ \varrho_{XAA'}$.

\begin{comment}
For
\be
\rho_{XAA'}=\sum_xp(x)\ket{x}\bra{x}_X\otimes\phi^x_{AA'},\nonumber
\ee
we have
\begin{align*}
\sigma_{XAB}&=\Psi_{A'\to B}\left(\rho_{XAA'}\right)\\
&=\sum_xp(x)\ket{x}\bra{x}_X\otimes\sigma^x_{AB},
\end{align*}
where
\be
\sigma^x_{AB}=\Psi_{A'\to B}\left(\phi^x_{AA'}\right).\nonumber
\ee
\end{comment}

\begin{comment}
The three entropic quantities of interest are of the form
\begin{align*}
I(AX;B)_\sigma&=S(B)_\sigma+\sum_xp(x)I(B\rangle A)_{\sigma^x},\\
I(A\rangle BX)_\sigma&=\sum_xp(x)I(A\rangle B)_{\sigma^x},\\
I(X;B)_\sigma&=S(B)_\sigma-\sum_xp(x)S(B)_{\sigma^x}.
\end{align*}

Since $ \Psi_{A'\to B}=\Psi_{A'\to B}\circ \Pi_{R\to R} $, consider the action of $ \Pi_{R\to R} $ on $ \phi^x_{AA'} $,
\be
\Pi_{R\to R}\left(\phi^x_{AA'}\right)=\sum_jp(j|x)\ket{j}\bra{j}_R\otimes\phi^{xj}_{AC},\nonumber
\ee
and thus
\be
\sigma_{XAB}=\sum_{x,j}p(x,j)\ket{x}\bra{x}_X\otimes\sigma^{xj}_{AB}\nonumber
\ee
where
\be
\sigma^{xj}_{AB}=\Psi_{A'\to B}\left(\ket{j}\bra{j}_R\otimes\phi^{xj}_{AC}\right)\nonumber
\ee
and
\be
p(x,j)\equiv  p(x)p(j|x)\label{eq:jprob}.
\ee
This means
\be
\sigma^{x}_{AB}=\sum_jp(j|x)\sigma^{xj}_{AB}.\label{eq:jdecomp}
\ee

Now consider the state
\be
\rho'_{X'AA'}=\sum_{x,j}p(x,j)\ket{x,j}\bra{x,j}_{X'}\otimes\ket{j}\bra{j}_R\otimes\phi^{xj}_{AC},\nonumber
\ee
with
\be
\sigma'_{X'AB}=\sum_{x,j}p(x,j)\ket{x,j}\bra{x,j}_{X'}\otimes\sigma^{xj}_{AB}.\nonumber
\ee

\end{comment}
\begin{enumerate}
	\item First consider $ I(A\rangle BX)_\sigma $.
	\begin{align}\label{eq:unitalcoherentinfo}
	I(A\rangle BX)_{\sigma}&=\sum_{x,j}p(x,j)I(A\rangle B)_{\sigma^{xj}}\nonumber\\
	&=\sum_{x,j}p(x)p(j|x)I(A\rangle B)_{\sigma^{xj}}\nonumber\\
	&\geq\sum_xp(x)I(A\rangle B)_{\varsigma^x}\nonumber\\
	&=I(A\rangle BX)_{\varsigma},
	\end{align}
	where the inequality is due to Eq.~(\ref{eq:jdecomp}) and the convexity of coherent information with respect to inputs.
	\item Now consider $ I(AX;B)_\sigma $. Similarly,
	\begin{align*}
	I(AX;B)_{\sigma}&=S(B)_{\sigma}+I(B\rangle AX)_{\sigma}\\
	&\geq S(B)_{\varsigma}+I(B\rangle AX)_\varsigma\\
	&=I(AX;B)_\varsigma,
	\end{align*}
	where the inequality is due to $ \sigma_B=\varsigma_B $ and Eq.~(\ref{eq:unitalcoherentinfo}).
	\item Finally consider $ I(X;B)_\sigma $. Writing $ \ket{x,j} $ as $ \ket{x}\ket{j} $, it can be shown
	\be
	I(X;B)_{\sigma}\geq I(X;B)_\varsigma\nonumber
	\ee
	using the data processing inequality when we apply the partial trace map $ \ket{x}\bra{x}\otimes\ket{j}\bra{j}\to\ket{x}\bra{x} $ to $ \sigma_{XB} $.
\end{enumerate}
\end{proof}
\begin{lemma}\label{lem:unitaldynamic}
The optimal trade-off surface of the 1-shot quantum dynamic capacity region of a unitally extended channel can always be achieved with $ \sigma_{XAB} $ such that $ S(B)_\sigma=\log(|B|) $. This extends similarly to higher shots.
\end{lemma}
\begin{proof}
Suppose $ \Phi_{RC\to B} $ is unitally extended from $ \Psi_{C\to B} $. Since a unitally extended channel $ \Phi_{RC\to B} $ is a partial cq channel, by Lemma \ref{lem:cqCQE}, we can consider states of the form
\be
\varrho_{XAA'}=\sum_{x,j}p(x,j)\ket{x,j}\bra{x,j}_X\otimes\ket{j}\bra{j}_R\otimes\phi^{xj}_{AC}.\nonumber
\ee
Let $\varsigma_{XAB} = \Phi_{RC\to B} (\varrho_{XAA'})$ with $A'\equiv RC$. Then
\be
\varsigma_{XAB}=\sum_{x,j}p(x,j)\ket{x,j}\bra{x,j}_X\otimes \varsigma^{xjj}_{AB},\nonumber
\ee
where $ \varsigma^{xjk}_{AB}= X(k)\varsigma^{xj}_{AB}X(k)^\dagger $ 
and $ \varsigma^{xj}_{AB}=\Psi_{C\to B}\left(\phi^{xj}_{AC}\right) $. 

We can construct another state of the form in Eq.~(\ref{eq:lem2input}):
\be \label{eq:lem3_cqinput}
\rho_{X'AA'}=\sum_{x,j,k}p(x,j,k)\ket{x,j,k}\bra{x,j,k}_{X'}\otimes\ket{k}\bra{k}_R\otimes\phi^{xj}_{AC},
\ee
where $ p(x,j,k)=p(x,j)/|R| $, and $\sigma_{X'AB}= \Phi_{RC\to B}(\rho_{X'AA'})$:
\be \nonumber
\sigma_{X'AB}=\sum_{x,j,k}p(x,j,k)\ket{x,j,k}\bra{x,j,k}_{X'}\otimes \sigma^{xjk}_{AB},
\ee
where $\sigma^{xjk}_{AB}= X(k)\sigma^{xj}_{AB}X(k)^\dagger $ and $ \sigma^{xj}_{AB}=\Psi_{C\to B}\left(\phi^{xj}_{AC}\right) $.  The state $\sigma_{X'AB}$ satisfies
\begin{align*}
S(B)_{\sigma}&= S\left(\sum_{x,j,k}p(x,j,k)\sigma^{xjk}_B\right)\\
&\geq \sum_{x,j}p(x,j)S\left(\frac{1}{|R|}\sum_k\sigma^{xjk}_B\right)\\
%&=\sum_{x,j}p(x,j)\log(|B|)
&=\log(|B|),
\end{align*}
where we've used the qudit twirl formula \cite{Wilde13}
\be\label{eq:twirl1}
\frac{1}{|R|}\sum_k\sigma^{xjk}_B=\frac{1}{|R|}\sum_k X(k)\sigma^{xj}_BX(k)^\dagger=\frac{1}{|B|}I_B.
\ee

One can verify that the dynamic capacity region with $\sigma_{X'AB}$ is larger than that with $\varsigma_{XAB}$ as follows:
%The three entropic quantities on $ \sigma_{XAB} $ are
%\begin{alignat*}{2}
%&I(AX;B)_\sigma&&=S(B)_\sigma+\sum_{x,j}p(x,j)I(B\rangle A)_{\sigma^{xjj}}\\
%& &&=S(B)_\sigma+\sum_{x,j}p(x,j)I(B\rangle A)_{\sigma^{xj}},\\
%&I(A\rangle BX)_\sigma&&=\sum_{x,j}p(x,j)I(A\rangle B)_{\sigma^{xjj}}\\
%& &&=\sum_{x,j}p(x,j)I(A\rangle B)_{\sigma^{xj}},\\
%&I(X;B)_\sigma&&=S(B)_\sigma-\sum_{x,j}p(x,j)S(B)_{\sigma^{xjj}}\\
%& &&=S(B)_\sigma-\sum_{x,j}p(x,j)S(B)_{\sigma^{xj}}.
%\end{alignat*}
%Now consider the state
%Now we can compare the three entropic quantities:
\begin{alignat}{2}
&I(A\rangle BX')_{\sigma}&&=\sum_{x,j,k}p(x,j,k)I(A\rangle B)_{\sigma^{xjk}}\nonumber\\
& &&=\sum_{x,j}p(x,j)I(A\rangle B)_{\varsigma^{xjj}}=I(A\rangle BX)_\varsigma \label{eq:lem3_1}\\
&I(AX';B)_{\sigma}&&=S(B)_{\sigma}+\sum_{x,j,k}p(x,j,k)I(B\rangle A)_{\sigma^{xjk}}\nonumber\\
& &&=\log(|B|)+\sum_{x,j}p(x,j)I(B\rangle A)_{\varsigma^{xjj}}\geq I(AX;B)_\varsigma\nonumber\\
&I(X';B)_{\sigma}&&=S(B)_{\sigma}-\sum_{x,j,k}p(x,j,k)S(B)_{\sigma^{xjk}}\\
& &&=\log(|B|)-\sum_{x,j}p(x,j)S(B)_{\varsigma^{xjj}}\geq I(X;B)_\varsigma.\label{eq:lem3_3}
\end{alignat}
The key property used in the above equations is, for any Heisenberg-Weyl operator $X(k)$,
\[
S(\sigma_B) = S(X(k)\sigma_B X(k)^\dagger).
\]
\end{proof}

\begin{prooflemma2}
Following from Lemma \ref{lem:unitaldynamic} and Eq.~(\ref{eq:lem3_cqinput}), we only need to consider states of the form
\begin{equation}\label{eq:Ninput}
\rho_{XAA'} =\sum_{m=0}^1p_m\ket{m}\bra{m}_M\otimes\rho^m_{XARC}
%=\sum_{m=0}^1\sum_{x,k} p(x,m,k)\ket{x,m,k}\bra{x,m,k}_{X}\otimes\\ \ket{m}\bra{m}_M\otimes\ket{k}\bra{k}_{R}\otimes\phi^{xm}_{AC}
\end{equation}
where $ p_m=\sum_{x,k}p(x,m,k)$ and
\[
\rho^m_{XARC}=\sum_{x,k}\frac{p(x,m,k)}{p_m}\ket{x,m,k}\bra{x,m,k}_X\otimes\ket{k}\bra{k}_R\otimes\phi^{xm}_{AC},
\]
with $ p(x,m,k)=p(x,m,k')$ for all $ k,k' $ and $m\in\{0,1\}$.
%\begin{align}
%&\rho_{XAA'}\label{eq:Ninput}\\
%=&\sum_{x,k}p(x,0,k)\ket{x,0,k}\bra{x,0,k}_{X}\otimes\ket{0}\bra{0}_M\otimes\ket{k}\bra{k}_R\otimes\phi^{x0}_{AC}\nonumber\\
%+&\sum_{x,k}p(x,1,k)\ket{x,1,k}\bra{x,1,k}_{X}\otimes\ket{1}\bra{1}_M\otimes\ket{k}\bra{k}_R\otimes\phi^{x1}_{AC},\nonumber
%\end{align}
%where $ p(x,0,k)=p(x,0,k') $ and $ p(x,1,k)=p(x,1,k') $ for all $ k,k' $.

The corresponding channel output is
\begin{align}\label{eq:lem4output}
&\sigma_{XAB}=\sum_{m=0}^1p_m\sigma^m_{XAB}
%=\sum_{m=0}^1\sum_{x,k}p(x,m,k)\ket{x,m,k}\bra{x,m,k}_X\otimes\sigma^{xmk}_{AB}
%\\
%+&\sum_{x,k}p(x,1,k)\ket{x,1,k}\bra{x,1,k}_X\otimes\sigma^{x1k}_{AB},
\end{align}
where
\be \label{eq:lem4output_m}
\sigma^m_{XAB}=\sum_{x,k}\frac{p(x,m,k)}{p_m}\ket{x,m,k}\bra{x,m,k}_X\otimes\sigma^{xmk}_{AB}
\ee
and
\begin{align*}
&\sigma_{AB}^{xmk}=X(k)\Psi^m_{C\to B}\left(\phi^{xm}_{AC}\right)X(k)^\dagger.
%\\
%&\sigma^{x1k}=X(k)\Psi^1_{C\to B}\left(\phi^{x1}_{AC}\right)X(k)^\dagger.
\end{align*}

\begin{comment}
Let $ p_m=\sum_{x,k}p(x,m,k)$, we can rewrite $ \rho_{XAA'} $ and $ \sigma_{XAB} $ as
\begin{align*}
\rho_{XAA'}&=\sum_{m=0}^1p_m\ket{m}\bra{m}_M\otimes\rho^m_{XARC}\\
\sigma_{XAB}&=\sum_{m=1}^1p_m\sigma^m_{XAB},
\end{align*}
where
\begin{alignat*}{2}
&\rho^m_{XARC}&&=\sum_{x,k}\frac{p(x,m,k)}{p_m}\ket{x,m,k}\bra{x,m,k}_X\otimes\ket{k}\bra{k}_R\otimes\phi^{xm}_{AC}\\
%&\rho^1_{XARC}&&=\sum_{x,k}\frac{p(x,1,k)}{p_1}\ket{x,1,k}\bra{x,1,k}_X\otimes\ket{k}\bra{k}_R\otimes\phi^{x1}_{AC}\\
&\sigma^m_{XAB}&&=\sum_{x,k}\frac{p(x,m,k)}{p_m}\ket{x,m,k}\bra{x,m,k}_X\otimes\sigma^{xmk}_{AB}.
%&\sigma^1_{XAB}&&=\sum_{x,k}\frac{p(x,1,k)}{p_1}\ket{x,1,k}\bra{x,1,k}_X\otimes\sigma^{x1k}_{AB},
\end{alignat*}
\end{comment}
Then all three of the entropic quantities evaluated on $\sigma_{XAB}$ in Eq.~(\ref{eq:lem4output}) can be decomposed to the corresponding ones evaluated on $ \sigma^m_{XAB} $ given in Eq.~(\ref{eq:lem4output_m}): 
\begin{align*}
I(A\rangle BX)_\sigma=&\sum_{m=0}^1\sum_{x,k}p(x,m,k)I(A\rangle B)_{\sigma^{xmk}}\\
=&\sum_{m=0}^1p_mI(A\rangle BX)_{\sigma^m}.
\end{align*}
Likewise,
\begin{align*}
I(AX;B)_{\sigma}=&\log(|B|)+\sum_{m=0}^1\sum_{x,k}p(x,m,k)I(B\rangle A)_{\sigma^{xmk}}\\
=&\sum_{m=0}^1 p_mI(AX;B)_{\sigma^m}
\end{align*}
and
\begin{align*}
I(X;B)_\sigma
%=&\log(|B|)-\sum_{x,k}\left(p(x,0,k)S(B)_{\sigma^{x0k}}+p(x,1,k)S(B)_{\sigma^{x1k}}\right)\\
=&\sum_{m=0}^1 p_mI(X;B)_{\sigma^m}.
\end{align*}
This means if we consider inputs of the form (\ref{eq:Ninput}), the triple rate for using $ \calN $ can always be expressed as a linear combination of the triple rates of $ \calN^0 $ and $ \calN^1 $. It is also clear that any linear combination is achievable by the time-sharing principle. Since using states of the form (\ref{eq:Ninput}) is optimal, we have
\begin{align*}
\calC^{(1)}_{CQE}\left(\calN\right)&=\bigcup_{0\leq p\leq 1}p\calC^{(1)}_{CQE}\left(\calN^0\right)+(1-p)\calC^{(1)}_{CQE}\left(\calN^1\right)\\
&=\rmconv\left(\calC^{(1)}_{CQE}\left(\calN^0\right),\calC^{(1)}_{CQE}\left(\calN^1\right)\right).
\end{align*}
Here, addition means Minkowski sum\footnote{For two sets of position vectors $ A $ and $ B $ in Euclidean space, their Minkowski sum $ A+B $ is obtained by adding each vector in $ A $ to each vector in $ B $, \textit{i.e.,} 
$
A+B=\{a+b|a\in A,b\in B\}
$ 
~\cite{Schneider93}.}.
Similarly, we have
\begin{align*}
\calC^{(1)}_{CQE}\bbr{\calN\otimes\calN}=&\rmconv\biggl(\calC^{(1)}_{CQE}\bbr{\calN^0\otimes\calN^0},\\
&\calC^{(1)}_{CQE}\left(\calN^0\otimes\calN^1\right),\calC^{(1)}_{CQE}\left(\calN^1\otimes\calN^1\right)\biggr).
\end{align*}
If the quantum dynamic capacity region is additive for $ \calN^0\otimes\Psi $, for any $ \Psi $, then
\be
\calC^{(1)}_{CQE}\left(\calN^0\otimes\calN^1\right)=\calC^{(1)}_{CQE}\left(\calN^0\right)+\calC^{(1)}_{CQE}\left(\calN^1\right).
\ee
In this case the 1-shot quantum dynamic capacity region for $ \calN\otimes\calN $ can be greatly simplified to
\be
\calC^{(1)}_{CQE}\left(\calN\otimes\calN\right)=\rmconv\left(2\calC_{CQE}\left(\calN^0\right),\calC^{(1)}_{CQE}\left(\calN^1\otimes\calN^1\right)\right). \nonumber
\ee
Similarly,
\begin{align*}
&\calC^{(1)}_{CQE}\bbr{\calN^{\otimes k}}\\
=&\rmconv\biggl(\calC_{CQE}\1\bbr{\bbr{\calN^1}^{\otimes k}},\calC_{CQE}\1\bbr{\calN^0\otimes\bbr{\calN^1}^{\otimes k-1}},\\
&\cdots,\calC_{CQE}\1\bbr{\bbr{\calN^0}^{\otimes k-1}\otimes\calN^1},\calC_{CQE}\1\bbr{\bbr{\calN^0}^{\otimes k}}\biggr).
\end{align*}
Each term $ \calC_{CQE}\1\bbr{\bbr{\calN^0}^{\otimes m}\otimes\bbr{\calN^1}^{\otimes k-m}} $, $0\leq m \leq k$, can be upper bounded as
\begin{align*}
&\calC_{CQE}\1\bbr{\bbr{\calN^0}^{\otimes m}\otimes\bbr{\calN^1}^{\otimes k-m}}\\
=&m\calC_{CQE}\bbr{\calN^0}+\calC_{CQE}\1\bbr{\bbr{\calN^1}^{\otimes k-m}}\\
\subseteq &m\calC_{CQE}\bbr{\calN^0}+(k-m)\calC_{CQE}\bbr{\calN^1}\\
\subseteq & k\rmconv\bbr{\calC_{CQE}\bbr{\calN^0},\calC_{CQE}\bbr{\calN^1}}.
\end{align*}
Here the second line follows from the addivity of the dynamic capacity region of $ \calN^0 $. The third line follows from the definition of $ \calC_{CQE} $. The fourth line follows from the definition of convex hull.
Thus $ \calC^{(1)}_{CQE}\bbr{\calN^{\otimes k}} $ can also be upper bounded as
\be
\calC^{(1)}_{CQE}\bbr{\calN^{\otimes k}}\subseteq k\rmconv\bbr{\calC_{CQE}\bbr{\calN^0},\calC_{CQE}\bbr{\calN^1}}.\nonumber
\ee
and
\begin{align*}
\calC_{CQE}\bbr{\calN}=&\overline{\bigcup_{k=1}^{\infty}\frac{1}{k}\calC^{(1)}_{CQE}\left(\calN^{\otimes k}\right)}\\
\subseteq &\overline{\rmconv\bbr{\calC_{CQE}\bbr{\calN^0},\calC_{CQE}\bbr{\calN^1}}}\\
=&\rmconv\bbr{\calC_{CQE}\bbr{\calN^0},\calC_{CQE}\bbr{\calN^1}}.
\end{align*}
The last equality follows because of the topology of the dynamic capacity region, as we show in Appendix \ref{Appendix:convexhullclosure}.

By a time-sharing protocol, it is obvious that
\be
\calC_{CQE}\bbr{\calN}\supseteq\rmconv\bbr{\calC_{CQE}\bbr{\calN^0},\calC_{CQE}\bbr{\calN^1}}.\nonumber
\ee
Hence
\be
\calC_{CQE}\left(\calN\right)=\rmconv\left(\calC_{CQE}\left(\calN^0\right),\calC_{CQE}\left(\calN^1\right)\right).\nonumber
\ee
\end{prooflemma2}

Note that unital extensions are not unique, and we only used the unitarity of Heisenberg-Weyl operators and the twirl formula Eq. (\ref{eq:twirl1}) in proving the above lemmas. Hence, as long as we have $ K $ unitaries $ \{U_k\}\in U(d) $ that satisfy the twirl formula
\be
\frac{1}{K}\sum_kU_k AU_k^\dagger=\textrm{Tr}(A)\frac{I}{d}\label{eq:twirl}
\ee
for any $ d\times d $ matrix $ A $, one has a valid unital extension, and lemmas \ref{lem:switchunital} and \ref{lem:unitaldynamic} will hold. \footnote{Note that we do not even require $ \calN^0 $ and $ \calN^1 $ to have the same unital extension. However, to ensure the input dimensions of $ \calN^0 $ and $ \calN^1 $ are the same, their unital extensions must involve the same number of unitaries. For this reason, we stick with the Heisenberg-Weyl operators most of the time.}

Moreover, unital extensions are preserved under tensor product of channels: if $ \Phi^1 $ is a unital extension of $ \Psi^1 $, and $ \Phi^2 $ is a unital extension of $ \Psi^2 $, then $ \Phi^1\otimes\Phi^2 $ is also a unital extension of $ \Psi^1\otimes\Psi^2 $. This follows from the fact that if $ \{U_j\}\in U(d_1) $ and $ \{V_k\}\in U(d_2) $ both satisfy Eq. ({\ref{eq:twirl}), then $ \{U_j\otimes V_k\}\in U(d_1d_2) $ also satisfies Eq. (\ref{eq:twirl}).
\section{Explicit Construction of Various Superadditivity Phenomena}\label{Sec:Construction}
With the tools developed in Sec \ref{Sec:Framework}, we can now explicitly construct channels that satisfy the superadditivity properties stated in Sec \ref{Sec:Summary}.
All our constructions utilize the switch channel idea. We always assume that $\calN$ is a switch channel of two unitally extended channels $ \calN^0 $ and $ \calN^1 $. Further, we assume that
\begin{enumerate}[label=\text{(U)}]
    \item \label{Prop:General} $\calN^0 $ has an additive dynamic capacity region, when tensored with another arbitrary channel.
\end{enumerate}
In this setting, we can use Lemma \ref{lem:switchunital} and its reduction to various single-resource and two-resource capacities.

In each construction, we first state the properties that $ \calN^0 $ and $ \calN^1 $ need to satisfy, in addition to Property \ref{Prop:General}. We then show how  the desired superadditivity of the switch channel $\calN$ follows from these properties. In the end, we explicitly construct channels that satisfy the properties we required.

%Since we need explicit constructions of $\calN^0$, we first summarize the channels that could be good candidates for $ \calN^0 $.

Before we start, we first propose two families of unital extended channels that satisfy \ref{Prop:General}. Many of our explicit constructions of $\calN^0$ will be chosen from these candidates. The first family comes from unital extensions of Hadamard channels. The following lemma shows that the dynamic capacity of the unitally extended Hadamard channels is also additive.

%First of all, it is known that Hadamard channels satisfy Property \ref{Prop:General} \cite{Bradler10}. Here, we establish that unital extensions of Hadamard channels also satisfy Property \ref{Prop:General}. Hence, they are good candidates for $ \calN^0 $.

%\begin{lemma}\cite{Bradler10}\label{lem:Hadamardadd}
%The dynamic capacity region is additive for $ \Psi^0 $ and any other channel $ \Psi^1 $, if $ \Psi^0 $ is a Hadamard channel.
%\end{lemma}
%The next lemma shows that the above result also holds for the unital  extension of a Hadamard channel. Its proof is similar to that of Lemma \ref{lem:Hadamardadd}, but we will provide it for completeness purpose.
\begin{lemma}\label{lem:unitalHadamard}
The dynamic capacity region is additive for $ \Phi^0 $ and any other channel $ \Psi^1 $, if $ \Phi^0 $ is a unital extension of a Hadamard channel $ \Psi^0 $.
\end{lemma}

The second family is unital extensions of classical channels.
%We also show that unital extensions of classical channels are good candidates for $ \calN^0 $.
\begin{lemma}\label{lem:unitalclassical}
If $ \Psi^0 $ is a classical channel, then the dynamic capacity region is additive for $ \Psi^0\otimes\Psi^1 $, for arbitrary $ \Psi^1 $. 
The same holds for a unital extension of a classical channel.
\end{lemma}

The proofs of the above lemmas are left to the Appendices, as they are not essential in understanding the construction.

\subsection{Additive C, Superadditive CE}\label{Subsec:CaddCEnonadd}
Here we review the original argument in \cite{ZhuZhuangShor17} and recast it in the current framework.

We use $ C_P\bbr{\calN} $ when we view $ C\bbr{\calN} $ as a function of the amount of entanglement assistance $ P $, where $ \bbr{C\bbr{\calN},P} $ are points on the CE trade-off curve of $ \calN $. When $ P=0 $, we return to the classical capacity $ \calC_C\bbr{\calN} $. When $ P $ is maximal, we arrive at the classical capacity with unlimited entanglement assistance $ C_E\bbr{\calN} $. $ C_P\1\bbr{\calN} $ denotes the 1-shot case.

We require $ \calN^0 $ and $ \calN^1 $ to have the following properties:

\begin{enumerate}[label=\text{(A\arabic*)}]
	\item \label{Prop:CaddCEnonadd1} $\calC_{C}\bbr{\calN^0}=\calC_{C}\bbr{\calN^1}$.
%	\be 
%	\calC_{C}\bbr{\calN^0}=\calC_{C}\bbr{\calN^1}.\nonumber
%	\ee 
	\item \label{Prop:CaddCEnonadd2}
	$ \calN^1 $ has a superadditive CE  trade-off capacity region, \textit{i.e.,},
	\be
	\calC_{CE}\bbr{\calN^1}\supsetneq\calC_{CE}\1\bbr{\calN^1},\nonumber
	\ee
	and $ \calC_{CE}\bbr{\calN^1} $ is strictly concave and superadditive at a boundary point of the trade-off region with entanglement consumption $ \bar{P} $.
	\item \label{Prop:CaddCEnonadd3} $\calC_{CE}\bbr{\calN^0}\subsetneq\calC_{CE}\bbr{\calN^1}$
%	\be
%	\calC_{CE}\bbr{\calN^0}\subsetneq\calC_{CE}\bbr{\calN^1}\nonumber
%	\ee
	in the sense the CE trade-off capacity region of $ \calN^0 $ is strictly smaller than that of $ \calN^1 $ when entanglement consumption is at $ \bar{P} $.
\end{enumerate}

In the $ C_P $ notation, property \ref{Prop:CaddCEnonadd2} means at $ P=\bar{P} $, $ C_P\bbr{\calN^1}>C_P\1\bbr{\calN^1} $ and $ C_P\bbr{\calN^1} $ is strictly concave\footnote{Here by saying a function $ f $ is strictly concave at $ y $, we mean $ f(y)>(1-p)f(v)+pf(w) $ for all $ v<y<w $ satisfying $ (1-p)v+pw=y $, with $ p\in(0,1) $.} in $ P $ at $ P=\bar{P} $ . Property \ref{Prop:CaddCEnonadd3} implies that $ C_P\bbr{\calN^0}<C_P\bbr{\calN^1} $ at $ P=\bar{P} $.

Note that these properties are weaker than the ones required in Ref. \cite{ZhuZhuangShor17}.

These three properties \ref{Prop:CaddCEnonadd1}-\ref{Prop:CaddCEnonadd3}, together with \ref{Prop:General}, will guarantee that (i) the classical capacity of $\calN$ is additive; and (ii) the CE trade-off capacity region of $\calN$ is superadditive at entanglement consumption rate $\bar{P}$.

Combining property \ref{Prop:CaddCEnonadd1} with \ref{Prop:General} yields statement (i):
\begin{align*}
\calC_C\bbr{\calN}=&\max\left\{\calC_C\bbr{\calN^0},\calC_C\bbr{\calN^1}\right\}=\calC_C\bbr{\calN^0}\\
=&\max\left\{\calC_C\1\bbr{\calN^0},\calC_C\1\bbr{\calN^1}\right\}=\calC_C\1\bbr{\calN}, 
\end{align*}
where Lemma~\ref{lem:switchunital} is used in the first equality.
%\textit{i.e.} the classical capacity of $ \calN $ is additive.

Property \ref{Prop:CaddCEnonadd3} ensures that
\begin{align}
\calC_{CE}\bbr{\calN}&=\rmconv\bbr{\calC_{CE}\bbr{\calN^0},\calC_{CE}\bbr{\calN^1}}\nonumber\\
&=\calC_{CE}\bbr{\calN^1}.\label{eq:CENconv}
\end{align}
Since
\be
\calC_{CE}\1\bbr{\calN}=\rmconv\bbr{\calC_{CE}\bbr{\calN^0},\calC_{CE}\1\bbr{\calN^1}},\nonumber
\ee
there exists $ P_0,P_1\geq0 $ and $ p\in [0,1] $ such that $ pP_0+(1-p)P_1=\bar{P} $ and
\be
C_{\bar{P}}\1\bbr{\calN}=pC_{P_0}\bbr{\calN^0}+(1-p)C_{P_1}\1\bbr{\calN^1}.\nonumber
\ee
Statement (ii) follows after considering three different cases.
\begin{enumerate}
	\item $ p=0 $.
	\be
	C_{\bar{P}}\1\bbr{\calN}=C_{\bar{P}}\1\bbr{\calN^1}<C_{\bar{P}}\bbr{\calN^1}=C_{\bar{P}}\bbr{\calN},\nonumber
	\ee
	where the inequality follows from the superadditivity part of property \ref{Prop:CaddCEnonadd2}. The second equality follows from Eq. (\ref{eq:CENconv}).
	\item $ 0<p<1 $.
	\begin{align*}
	C_{\bar{P}}\1\bbr{\calN}&=pC_{P_0}\1\bbr{\calN^0}+(1-p)C_{P_1}\1\bbr{\calN^1}\\
	&\leq pC_{P_0}\bbr{\calN^1}+(1-p)C_{P_1}\bbr{\calN^1}\\
	&<C_{\bar{P}}\bbr{\calN^1}=C_{\bar{P}}\bbr{\calN}
	\end{align*}
	where the first inequality follows from Property \ref{Prop:CaddCEnonadd3}. The second inequality follows from the strict concavity part of property \ref{Prop:CaddCEnonadd2}. The last equality follows from Eq. (\ref{eq:CENconv}).
	\item $ p=1 $. Then
	\be
	C_{\bar{P}}\1\bbr{\calN}=C_{\bar{P}}\bbr{\calN^0}<C_{\bar{P}}\bbr{\calN^1}=C_{\bar{P}}\bbr{\calN}.\nonumber
	\ee
	Here the first equality follows from additivity of the CE trade-off capacity region for $ \calN^0 $. The inequality follows from property \ref{Prop:CaddCEnonadd3}. The last equality follows from Eq. (\ref{eq:CENconv}).
\end{enumerate}

\subsubsection*{Explicit Construction of $\calN$}
We quote the following property about concave functions \cite{Rudin76}:
A concave function $ u(y) $ is continuous, differentiable from the left and from the right. The derivative is decreasing, \textit{i.e.,} for $ x<y $ we have $ u'(x-)\geq u'(x+)\geq u'(y-)\geq u'(y+) $.
We use ``$ \pm $'' to denote the right and left derivatives when needed.

We first construct $ \calN^1 $. Choose  $ \Psi^{\text{ro}} $ to be a random orthogonal channel with a subadditive minimum output entropy, and  $ \Psi^{\text{ro}} $ has input dimension $ N $. This is unitally extended to $ \Phi^{\text{ro}} $.

Due to Lemma~\ref{lem:cqCQE}, the useful entanglement assistance is at most $ \log(N) $. Thus we restrict to $ 0\leq P\leq \log(N) $.

Let
\be
\epsilon=\calC_C\bbr{\Phi^{\text{ro}}}-\calC_C\1\bbr{\Phi^{\text{ro}}}>0\label{eq:gap}.
\ee
Since 
\be
C\1_P\bbr{\Phi^{\text{ro}}}\leq \calC_C\1\bbr{\Phi^{\text{ro}}}+P\label{eq:CPub},
\ee
\be
C_E\bbr{\Phi^{\text{ro}}}\leq \calC_C\bbr{\Phi^{\text{ro}}}+\log(N)-\epsilon.\nonumber
\ee
This implies $ dC_P\bbr{\Phi^{\text{ro}}}/dP $ cannot always be 1. Thus there exists $ \bar{P}\in[0,\log(N)) $ such that
\be
dC_P\bbr{\Phi^{\text{ro}}}/dP=1,~~\forall~ 0\leq P\leq\bar{P}\nonumber
\ee
and
\be
dC_P\bbr{\Phi^{\text{ro}}}/dP<1,~~\forall P>\bar{P}.\nonumber
\ee

Next we discuss different cases of $ \bar{P} $.
\begin{enumerate}
	\item $ \bar{P}>0 $. Then $ C_P\bbr{\Phi^{\text{ro}}} $ is strictly concave at $ \bar{P} $. Furthermore, $ C_{\bar{P}}\bbr{\Phi^{\text{ro}}}-C\1_{\bar{P}}\bbr{\Phi^{\text{ro}}}\geq \epsilon $ since  $ C_{\bar{P}}\bbr{\Phi^{\text{ro}}}=\calC_C\bbr{\Phi^{\text{ro}}}+\bar{P} $
	but $ C\1_{\bar{P}}\bbr{\Phi^{\text{ro}}}\leq \calC_C\1\bbr{\Phi^{\text{ro}}}+\bar{P} $. Thus $ \calN^1=\Phi^{\text{ro}} $ satisfies \ref{Prop:CaddCEnonadd2}.
	\item $ \bar{P}=0 $. Let $ \calN^1=\Phi^{\text{ro}}\otimes\Phi^{\text{dph}}_\eta $, where $ \Phi^{\text{dph}}_\eta $ is the unital extension of the qubit dephasing channel. 
		
	Since $	dC_P\bbr{\Phi^{\text{ro}}}/dP|_{0_+}$$<1$,
	choose $ \eta>0 $ small such that $
	dC_P\left(\Phi^{\text{dph}}_\eta\right)/dP|_{1_-}$$>dC_P\bbr{\Phi^{\text{ro}}}/dP|_{0_+} $. This is possible, as  $C_P\bbr{\Phi^{\text{dph}}_\eta}=C_P\bbr{\Psi^{\text{dph}}_\eta} $ and $ dC_P\bbr{\Psi^{\text{dph}}_\eta}/dP|_{1_-}\to 1 $ as $ \eta\to 0 $.
	This ensures that when $ 0<P\leq 1$,
	\be
	C_P\bbr{\calN^1}=\calC_C\bbr{\Phi^{\text{ro}}}+C_P\left(\Phi^{\text{dph}}_\eta\right)\label{eq:CPadd},
	\ee
	where we've also used Lemma \ref{lem:unitalHadamard}.
	
	For $ \Phi^{\text{dph}}_\eta $, it can be shown that $ C_P\left(\Phi^{\text{dph}}_\eta\right) $ is strictly concave in $ P $ when $ \eta<1/2 $ (see Appendix \ref{Appendix:dephasing}). Hence $ C_P\bbr{\calN^1} $ is also strictly concave with respect to $ P $, for $ 0<P\leq 1 $. Also, when $ P<\epsilon $,
	\begin{align*}
	C_P\bbr{\calN^1}&> \calC_C\bbr{\Phi^{\text{ro}}}+\calC_C\left(\Phi^{\text{dph}}_\eta\right)\\
	&>\calC_C^{(1)}\bbr{\Phi^{\text{ro}}}+\calC_C\bbr{\Phi^{\text{dph}}_\eta}+P\geq C\1_P\bbr{\calN^1}.
	\end{align*}
	Here the first inequality comes from Eq. (\ref{eq:CPadd}) and $ C_P\bbr{\Phi^{\text{dph}}_\eta}>\calC_C\bbr{\Phi^{\text{dph}}_\eta} $ when $ P>0 $. The second inequality comes from our assumption $ P<\epsilon $ and Eq. (\ref{eq:gap}). The last inequality comes from Eq. (\ref{eq:CPub}).\\		
	This ensures that $ C_P\bbr{\calN^1} $ is superadditive. Thus when $ 0<P<\min\{1,\epsilon\}$, $ C_P\bbr{\calN^1} $ is strictly concave and superadditive.
\end{enumerate}

For $ \calN^0 $, as long as it is a unital extension of a classical channel with $ \calC_C\bbr{\calN^0}=\calC_C\bbr{\calN^1} $, it will automatically satisfy property \ref{Prop:CaddCEnonadd3}.

\subsection{Additive C and Q, Superadditive CE}\label{Subsec:CnQaddCEnonadd}
In Section \ref{Subsec:CaddCEnonadd}, we constructed a channel $ \calN $ with an additive classical capacity, but a superadditive CE trade-off capacity region. It's unclear if our construction $ \calN $ has an additive quantum capacity. To extend the argument, we need to make some modifications to the original construction.

In addition to properties \ref{Prop:CaddCEnonadd1}-\ref{Prop:CaddCEnonadd3}, the channels $\calN^0$ and $\calN^1$ need to satisfy
\begin{enumerate}[label=\text{(B\arabic*)}]
	\item \label{Prop:CnQaddCEnonadd} $\calC_{Q}\bbr{\calN^0}\geq\calC_{Q}\bbr{\calN^1}.$
\end{enumerate}
\begin{comment}
Firstly, recall that our argument in Section \ref{Subsec:CaddCEnonadd} requires $ \calN^0 $ and $ \calN^1 $ to have the following properties:

\begin{enumerate}[label=\text{(B\arabic*)}]
	\item \label{Prop:CnQaddCEnonadd1}
	\be 
	\calC_{C}\bbr{\calN^0}=\calC_{C}\bbr{\calN^1}.\nonumber
	\ee 
	\item \label{Prop:CnQaddCEnonadd2}
	$ \calN^1 $ has a superadditive CE  trade-off region, meaning
	\be
	\calC_{CE}\bbr{\calN^1}\supsetneq\calC_{CE}\1\bbr{\calN^1}.\nonumber
	\ee
	$ \calC_{CE}\bbr{\calN^1} $ is strictly concave and superadditive at a boundary point with entanglement consumption $ P $.
	\item \label{Prop:CnQaddCEnonadd3}
	\be
	\calC_{CE}\bbr{\calN^0}\subsetneq\calC_{CE}\bbr{\calN^1}\nonumber
	\ee
	in the sense the CE trade-off region of $ \calN^0 $ is strictly smaller than that of $ \calN^1 $ when entanglement consumption is at $ P $.
\end{enumerate}

\end{comment}

%To construct a channel with an additive quantum capacity, we need the following additional properties on $ \calN^0 $ and $ \calN^1 $
%\be
%\calC_{Q}\bbr{\calN^0}\geq\calC_{Q}\bbr{\calN^1}\label{eq:CEnonaddQadd}
%\ee
This ensures that the quantum capacity of $ \calN $ is also additive:
\begin{align*}
\calC_{Q}\bbr{\calN}&=\max\left\{\calC_{Q}\bbr{\calN^0},\calC_{Q}\bbr{\calN^1}\right\}=\calC_{Q}\bbr{\calN^0}\\
&=\max\left\{\calC_{Q}\1\bbr{\calN^0},\calC_{Q}\1\bbr{\calN^1}\right\}=\calC_{Q}\bbr{\calN}.
\end{align*}
%\textit{i.e.} the quantum capacity of $ \calN $ is additive.

\subsubsection*{Explicit Construction of $\calN$}
We take the channels $ \calN^0 $ and $ \calN^1 $ that were constructed in Sec \ref{Subsec:CaddCEnonadd}, and compare their quantum capacities. Since $ \calC_Q\bbr{\calN^0}=0 $, we can only have $ \calC_Q\bbr{\calN^0}\leq \calC_Q\bbr{\calN^1} $. If
\be
\calC_Q\bbr{\calN^0}= \calC_Q\bbr{\calN^1},\nonumber
\ee
then \ref{Prop:CnQaddCEnonadd} is automatically satisfied. Hence we will focus on the case where
\be
\calC_Q\bbr{\calN^0}< \calC_Q\bbr{\calN^1}.\nonumber
\ee
In this case, we call these two channels $ \varPhi ^0 $ and $ \varPhi^1 $ respectively. We will construct two new channels $ \calN^0 $ and $ \calN^1 $ that satisfy properties \ref{Prop:CaddCEnonadd1}-\ref{Prop:CaddCEnonadd3} and \ref{Prop:CnQaddCEnonadd}.

We will use the qubit dephasing channel and $ 1\to N $ cloning channel. To make the argument work, we will modify them in the following manner.

For the $ 1\to N $ cloning channel $ \Psi^{1\to N} $, we always tensor an appropriate classical channel, such that the resulting channel has its classical capacity equal to 1, and the output dimension is the same as the input dimension. We denote the resulting channel $ \Psi^N $.

For the dephasing channel, we will tensor a complete depolarizing channel, so that its input and output dimensions match those of $ \Psi^N $. Since tensoring a complete depolarizing channel does not modify the dynamic capacity region of the qubit dephasing channel, we will continue using $ \Psi^{\text{dph}}_\eta $ to denote it.

Based on results in Ref.~\cite{Bradler10}, we can obtain the trade-off capacities of the qubit dephasing channel $ \Psi^{\text{dph}}_\eta $ and modified $ 1\to N $ cloning channel $ \Psi^N $.
We observe that for $ \eta=0.2 $ and $ N=15 $, their trade-off capacities satisfy the following properties (see Fig.~\ref{CE_CQ})
\be
\calC_{Q}\bbr{\Psi^{\text{dph}}_\eta}>\calC_{Q}\bbr{\Psi^N}.\nonumber
\ee
and
\be\label{eq:CEdph1toN}
\calC_{CE}\bbr{\Psi^{\text{dph}}_\eta}\subsetneq\calC_{CE}\bbr{\Psi^N},
\ee
in the sense that $ \Psi^N $ achieves a strictly better classical communication rate than $ \Psi^{\text{dph}}_\eta $, if we have any non-zero amount of entanglement assistance. In the $ C_P $ notation, it means $ C_P\bbr{\Psi^{\text{dph}}_\eta}<C_P\bbr{\Psi^N} $ for all $ P>0 $.
\begin{figure}
    \subfigure[]{
	\includegraphics[width=0.22\textwidth]
	{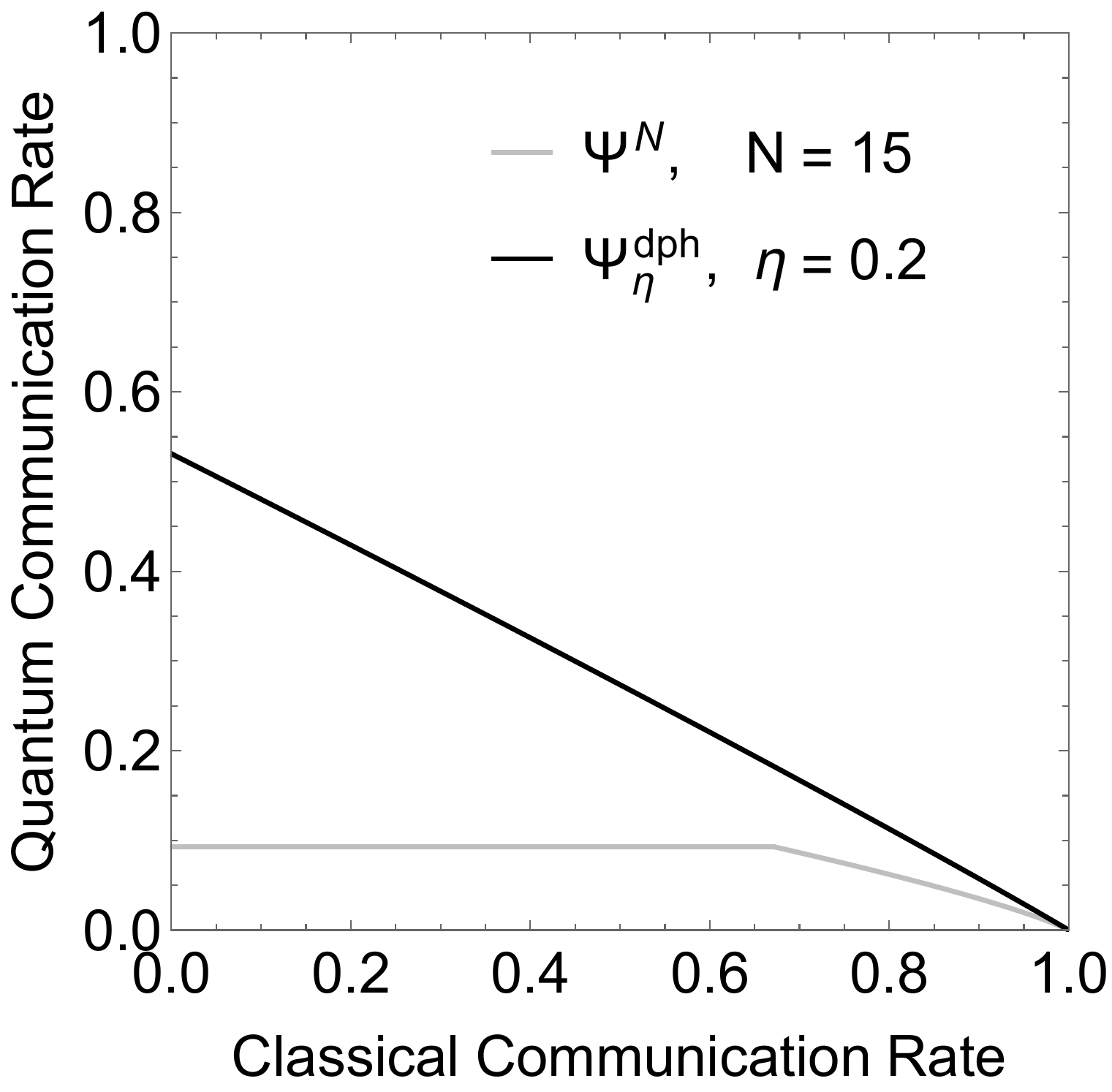}
		\label{fig:C_Q_CE_1}
	}
	\subfigure[]{
	\includegraphics[width=0.22\textwidth]
	{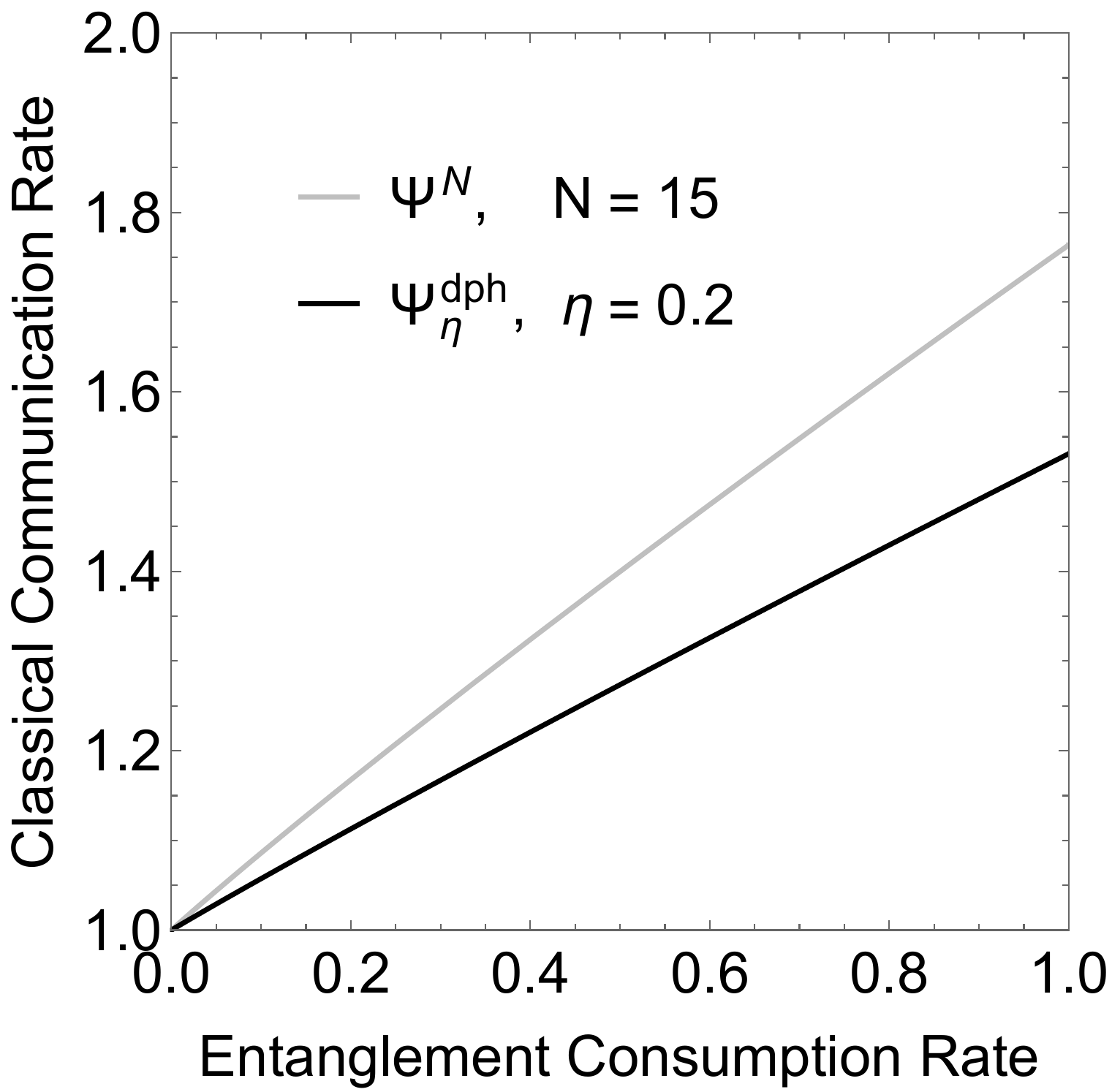}
		\label{fig:C_Q_CE_2}
	}
	\caption{ Comparison of trade-off curves between qubit dephasing channel $ \Psi^{\text{dph}}_\eta $ and modified $ 1\to N $ cloning channel $ \Psi^N $, when $ \eta=0.2 $ and $ N=15 $. (a) CQ trade-off. (b) CE trade-off. 
	\label{CE_CQ}
	}
\end{figure}

Since unital extensions do not change the CE and CQ trade-off capacity regions of these two channels (see Appendix \ref{Appendix:dephasing}), the above properties hold if we replace $ \Psi^{\text{dph}}_\eta $ and $ \Psi^N $ by their unital extensions $ \Phi^{\text{dph}}_\eta $ and $ \Phi^N $ respectively.

Since
\be
\calC_{Q}\bbr{\Phi^{\text{dph}}_\eta}>\calC_{Q}\bbr{\Phi^N},\nonumber
\ee
let $ n $ be large enough so that
\be
n\calC_{Q}\bbr{\Phi^{\text{dph}}_\eta}+\calC_Q\bbr{\varPhi^0}\geq n\calC_{Q}\bbr{\Phi^N}+\calC_Q\bbr{\varPhi^1}.\nonumber
\ee
Define
\be
\calN^0=\bbr{\Phi^{\text{dph}}_\eta}^{\otimes n}\otimes \varPhi^0 \nonumber
\ee
and
\be
\calN^1=\bbr{\Phi^N}^{\otimes n}\otimes \varPhi^1.\nonumber
\ee
Our choice of $ n $ ensures that
\be
\calC_Q\bbr{\calN^0}\geq \calC_Q\bbr{\calN^1}.\nonumber
\ee

We also need to ensure our newly constructed $ \calN^0 $ and $ \calN^1 $ still satisfy properties \ref{Prop:CaddCEnonadd1}-\ref{Prop:CaddCEnonadd3}.

As
\be
\calC_C\bbr{\Phi^{\text{dph}}_\eta}=\calC_C\bbr{\Phi^N}=1\nonumber
\ee
and
\be
\calC_C\bbr{\varPhi^0}=\calC_C\bbr{\varPhi^1},\nonumber
\ee
we immediately have
\be
\calC_C\bbr{\calN^0}=\calC_C\bbr{\calN^1}\nonumber
\ee
and property \ref{Prop:CaddCEnonadd1} is satisfied.

The CE trade-off curve of $ \Psi^{1\to N} $ is strictly concave for $ N\neq1 $ \cite{Bradler10}, hence property \ref{Prop:CaddCEnonadd2} is also satisfied for $ \calN^1 $.

Property \ref{Prop:CaddCEnonadd3} is satisfied due to Eq. (\ref{eq:CEdph1toN}).

\subsection{Additive Q, Superadditive CQ}\label{Subsec:QaddCQnonadd}

We require $ \calN^0 $ and $ \calN^1 $ to have the following properties:
\begin{enumerate}[label=\text{(C\arabic*)}]
	\item \label{Prop:QaddCQnonadd1} $	\calC_{Q}\bbr{\calN^0}\geq\calC_{Q}\bbr{\calN^1}.$	
	\item \label{Prop:QaddCQnonadd2}
	$\calC_{C}\bbr{\calN^1}>\calC_{C}\1\bbr{\calN^1}.$
	\item \label{Prop:QaddCQnonadd3}
    $\calC_{C}\bbr{\calN^0}<\calC_{C}\bbr{\calN^1}.$
\end{enumerate}

These properties \ref{Prop:QaddCQnonadd1}-\ref{Prop:QaddCQnonadd3} will allow us to show that (i) $
\calC_{Q}\bbr{\calN}=\calC_{Q}\1\bbr{\calN}$; and (ii) 
$\calC_{CQ}\bbr{\calN}\supsetneq\calC_{CQ}\1\bbr{\calN}.$

Statement (i) follows from property \ref{Prop:QaddCQnonadd1} and \ref{Prop:General} that $ \calN^0 $ has an additive quantum capacity:
\begin{align*}
\calC_{Q}\bbr{\calN}&=\max\left\{\calC_{Q}\bbr{\calN^0},\calC_{Q}\bbr{\calN^1}\right\}=\calC_{Q}\bbr{\calN^0}\\
&=\max\left\{\calC_{Q}\1\bbr{\calN^0},\calC_{Q}\1\bbr{\calN^1}\right\}=\calC_{Q}\1\bbr{\calN}.
\end{align*}
%\textit{i.e.} the quantum capacity of $ \calN $ is additive.

Properties \ref{Prop:QaddCQnonadd2} and \ref{Prop:QaddCQnonadd3} together ensure
\begin{align*}
\calC_{C}\bbr{\calN}&=\max\left\{\calC_{C}\bbr{\calN^0},\calC_{C}\bbr{\calN^1}\right\}=\calC_{C}\bbr{\calN^1}\\
&>\max\left\{\calC_{C}\1\bbr{\calN^0},\calC_{C}\1\bbr{\calN^1}\right\}=\calC_{C}\1\bbr{\calN},
\end{align*}
\textit{i.e.,} the classical capacity of $ \calN $ is superadditive; hence statement (ii) follows.

\subsubsection*{Explicit Construction of $\calN$}
Next we construct $ \calN^0 $ and $ \calN^1 $ that satisfy the above properties.

Let $ \Psi^{\text{ro}} $ be a random orthogonal channel, such that its unital extension has a superadditive classical capacity. For convenience, we also assume $ \Psi^{\text{ro}} $ has the input dimension $ N=2^n $. Choose $ \eta $ for the qubit dephasing channel $ \Psi^{\text{dph}}_\eta $ such that $ \calC_{Q}\bbr{\Psi^{\text{ro}}}+\calC_{Q}\bbr{\Psi^{\text{dph}}_\eta}=m $ for some integer $ m $.

Define
\be
\calN^1=\Phi^{\text{ro}}\otimes\Phi^{\text{dph}}_\eta,\nonumber
\ee
where $ \Phi^{\text{ro}} $ is a unital extension of $ \Psi^{\text{ro}} $ and $ \Phi^{\text{dph}}_\eta $ is a unital extension of $ \Psi^{\text{dph}}_\eta $. $ \calN^1 $ has the property that its quantum capacity is $ \calC_{Q}\bbr{\calN^1}=m $, whereas its classical capacity is superadditive, and greater than $ m $.

Define
\be
\calN^0=\bbr{\Phi^{\calI}}^{\otimes m}\otimes\bbr{\Phi^{\text{dpo}}_1}^{\otimes n+1-m},\nonumber
\ee
where $ \Phi^{\calI} $ is a unital extension of the noiseless qubit channel, and $ \Phi^{\text{dpo}}_1 $ is a unital extension of the complete qubit depolarizing channel.

It's clear that $ \calN^0 $ has its classical and quantum capacity as $ \calC_{C}\bbr{\calN^0}=\calC_{Q}\bbr{\calN^0}=m $, thus fulfiling the properties \ref{Prop:QaddCQnonadd1} and \ref{Prop:QaddCQnonadd3} above.
\subsection{Additive C and Q, Superadditive CQ}\label{Subsec:CnQaddCQnonadd}

We require $ \calN^0 $ and $ \calN^1 $ to satisfy the following properties:
\begin{enumerate}[label=\text{(D\arabic*)}]
	\item \label{Prop:CnQaddCQnonadd1} 
	$\calC_{C}\bbr{\calN^0}\leq\calC_{C}\bbr{\calN^1}=\calC_{C}\1\bbr{\calN^1}\nonumber
	$ and $
	\calC_{Q}\bbr{\calN^0}=\calC_{Q}\bbr{\calN^1}.$
	\item \label{Prop:CnQaddCQnonadd2} 
	$ \calN^1 $ has a superadditive CQ  trade-off capacity region, meaning
	\be
	\calC_{CQ}\bbr{\calN^1}\supsetneq\calC_{CQ}\1\bbr{\calN^1}.\nonumber
	\ee
	$ \calC_{CQ}\bbr{\calN^1} $ is strictly concave and superadditive at a boundary point with classical communication rate $ \bar{C} $.
	\item \label{Prop:CnQaddCQnonadd3}
	$\calC_{CQ}\bbr{\calN^0}\subsetneq\calC_{CQ}\bbr{\calN^1}$
	in the sense the CQ trade-off capacity region of $ \calN^0 $ is strictly smaller than that of $ \calN^1 $ when classical communication rate is at  $ \bar{C} $.
\end{enumerate}

With these properties, we can show that (i)
$\calC_C\left(\calN\right)=\calC^{(1)}_C\left(\calN\right)$; (ii) $
\calC_Q\left(\calN\right)=\calC^{(1)}_Q\left(\calN\right)$; and (iii) $ \calC_{CQ}\left(\calN\right)\supsetneq \calC^{(1)}_{CQ}\left(\calN\right).$

We'll focus on the CQ trade-off curve. Same as in Section \ref{Subsec:CaddCEnonadd}, we use a simplified notation $ Q_C\bbr{\calN} $ when we view $ Q\bbr{\calN} $ as a function of $ C\bbr{\calN} $. In the 1-shot scenario, it is denoted by $ Q\1_C\bbr{\calN} $. We'll show there exists $ \bar{C}\neq 0 $ such that $ Q_{\bar{C}}\bbr{\calN}>Q_{\bar{C}}\1\bbr{\calN} $.

In the $ Q_C $ notation, property \ref{Prop:CnQaddCQnonadd2} means at $ C=\bar{C} $, $ Q_C\bbr{\calN^1}>Q_C\1\bbr{\calN^1} $ and $ Q_C\bbr{\calN^1} $ is strictly concave in $ C $ at $ C=\bar{C} $. Property \ref{Prop:CnQaddCQnonadd3} implies that $ Q_C\bbr{\calN^0}<Q_C\bbr{\calN^1} $ at $ C=\bar{C} $.

Property \ref{Prop:CnQaddCQnonadd1} and \ref{Prop:General} that $ \calN^0 $ has an additive quantum capacity ensure that
\begin{align*}
\calC_{C}\bbr{\calN}&=\max\left\{\calC_{C}\bbr{\calN^0},\calC_{C}\bbr{\calN^1}\right\}=\calC_{C}\bbr{\calN^1}\\
&=\max\left\{\calC_{C}\1\bbr{\calN^0},\calC_{C}\1\bbr{\calN^1}\right\}=\calC_{C}\1\bbr{\calN}
\end{align*}
and
\begin{align*}
\calC_{Q}\bbr{\calN}&=\max\left\{\calC_{Q}\bbr{\calN^0},\calC_{Q}\bbr{\calN^1}\right\}=\calC_{Q}\bbr{\calN^0}\\
&=\max\left\{\calC_{Q}\1\bbr{\calN^0},\calC_{Q}\1\bbr{\calN^1}\right\}=\calC_{Q}\1\bbr{\calN},
\end{align*}
\textit{i.e.,} $ \calN $ has an additive classical and quantum capacity.

By property \ref{Prop:CnQaddCQnonadd3}, we have
\begin{align}
\calC_{CQ}\left(\calN\right)&=\rmconv\left(\calC_{CQ}\left(\calN^0\right),\calC_{CQ}\left(\calN^1\right)\right)\nonumber\\
&=\calC_{CQ}\left(\calN^1\right).\label{eq:CQNconv}
\end{align}
Since
\be
\calC_{CQ}^{(1)}\left(\calN\right)=\rmconv\left(\calC_{CQ}\left(\calN^0\right),\calC_{CQ}^{(1)}\left(\calN^1\right)\right),\nonumber
\ee
there exists $ C_0,C_1 $ and $ p\in [0,1] $ such that $ pC_0+(1-p)C_1=\bar{C} $ and
\be
Q_{\bar{C}}\1\bbr{\calN}=pQ_{C_0}\bbr{\calN^0}+(1-p)Q_{C_1}\1\bbr{\calN^1}.\nonumber
\ee
Now consider three different cases.
\begin{enumerate}
	\item $ p=0 $.
	\be
	Q_{\bar{C}}\1\bbr{\calN}=Q_{\bar{C}}\1\bbr{\calN^1}<Q_{\bar{C}}\bbr{\calN^1}= Q_{\bar{C}}\bbr{\calN},\nonumber
	\ee
	where the inequality follows from property \ref{Prop:CnQaddCQnonadd2}. The second equality follows from Eq. (\ref{eq:CQNconv}).
	\item $ 0<p<1 $.
	\begin{align*}
	Q_{\bar{C}}\1\bbr{\calN}=&pQ_{C_0}\bbr{\calN^0}+(1-p)Q_{C_1}\1\bbr{\calN^1}\\
	\leq &pQ_{C_0}\bbr{\calN^1}+(1-p)Q_{C_1}\bbr{\calN^1}\\
	<&Q_{\bar{C}}\bbr{\calN^1}= Q_{\bar{C}}\bbr{\calN}.
	\end{align*}
	Here the first inequality follows from the definition of $ \calC_{CQ} $ and property \ref{Prop:CnQaddCQnonadd3}. The second inequality follows from the strict concavity part of property \ref{Prop:CnQaddCQnonadd2}. The last equality follows from Eq. (\ref{eq:CQNconv}).
	\item $ p=1 $. Then
	\be
	Q_{\bar{C}}\1\bbr{\calN}=Q_{\bar{C}}\bbr{\calN^0}<Q_{\bar{C}}\bbr{\calN^1}=Q_{\bar{C}}\bbr{\calN}.\nonumber
	\ee
	Here the inequality follows from property \ref{Prop:CnQaddCQnonadd3}. The last equality follows from Eq. (\ref{eq:CQNconv}).
\end{enumerate}
Hence statement (iii) follows.

\subsubsection*{Explicit Construction}
Now we explicitly construct $ \calN^0 $ and $ \calN^1 $.

Choose $ p $ such that the qubit depolarizing channel $ \Psi^{\text{dpo}}_p $ is known to have a superadditive quantum capacity. Consider its unital extension $ \Phi^{\text{dpo}}_p $. Note that the gradient $ dQ_C\bbr{\Phi^{\text{dpo}}_p}/dC $ of the CQ trade-off curve cannot always stay at 0 for the choice of $ \Psi^{\text{dpo}}_p $ with a positive quantum capacity. %due to the end point $ \bbr{\calC_C\bbr{\Phi^{\text{dpo}}_p},0} $ of the CQ trade-off curve. 
It means there exists $ 0\leq\bar{C}< \calC_C\bbr{\Phi^{\text{dpo}}_p} $ such that
\be
dQ_C\bbr{\Phi^{\text{dpo}}_p}/dC=0,~~\forall 0\leq C\leq \bar{C}_-\label{eq:CQgradient}
\ee
and
\be
dQ_C\bbr{\Phi^{\text{dpo}}_p}/dC<0,~~\forall \bar{C}_+\leq C\leq \calC_C\bbr{\Phi^{\text{dpo}}_p}.\nonumber
\ee
\begin{enumerate}
	\item $ \bar{C}>0 $. In this case, we know $ Q_C\bbr{\Phi^{\text{dpo}}_p} $ is strictly concave at $ \bar{C} $.
	Also	
\be
	Q_{\bar{C}}\bbr{\Phi^{\text{dpo}}_p}=Q_0\bbr{\Phi^{\text{dpo}}_p}>Q_0\1\bbr{\Phi^{\text{dpo}}_p}\geq Q_{\bar{C}}\1\bbr{\Phi^{\text{dpo}}_p}.\nonumber
	\ee
	Here the equality follows from Eq. (\ref{eq:CQgradient}). The first inequality follows because $ \Psi^{\text{dpo}}_p $ has a superadditive quantum capacity, as both $ \calC_Q $ and $ \calC_Q\1$ remain unchanged after a unital extension, and $ Q_C $ reduces to the quantum capacity at $ C=0 $. The second inequality follows as the rate of quantum communication along the CQ trade-off curve must not exceed the quantum capacity.
	
	Choose the noise parameter $ \eta $ for the qubit dephasing channel $ \Psi^{\text{dph}}_\eta $ appropriately such that
	\be
	\calC_Q\bbr{\Psi^{\text{dph}}_\eta}=1-\calC_Q\bbr{\Psi^{\text{dpo}}_p}.\nonumber
	\ee
	Define
	\be
	\calN^1=\Phi^{\text{dpo}}_p\otimes\Phi^{\text{dph}}_\eta.\nonumber
	\ee
	It's clear that $ \calN^1 $ is a unitally extended channel of $ \Psi^{\text{dpo}}_p\otimes\Psi^{\text{dph}}_\eta $ and has $ \calC_Q\bbr{\calN^1}=\calC_Q\bbr{\Psi^{\text{dpo}}_p\otimes\Psi^{\text{dph}}_\eta}=1 $. The CQ trade-off curve is strictly concave and superadditive at $ \bar{C} $.
	The corresponding $ \Psi^0 $ is
	\be
	\Psi^0=\calI\otimes\Psi^{\text{dpo}}_1,\nonumber
	\ee
	\textit{i.e.,} a noiseless channel tensor a complete qubit depolarizing channel. $ \calN^0 $ is a unital extension of $ \Psi^0 $. %\textcolor{red}{We have to revise the beginning of this Sec, since it reads like we only choose Hadamart channels and identity channels for $\calN^0$}.
	\item $ \bar{C}=0 $. Choose $ \eta_1 $ close to 1/2 such that 
	\be
	\frac{dQ_C\left(\Phi^{\text{dph}}_{\eta_1}\right)}{dC}\biggr|_{\calC_C\bbr{\Psi^{\text{dph}}_{\eta_1}}_-}>\frac{dQ_C\bbr{\Phi^{\text{dpo}}_p}}{dC}\biggr|_{0_+}.\nonumber
	\ee
	Let
	\be
	\calN^1=\Phi^{\text{dph}}_{\eta_1}\otimes\Phi^{\text{dpo}}_p\otimes\Phi^{\text{dph}}_{\eta_2},\nonumber
	\ee
	where $ \eta_2 $ is chosen such that
	\begin{align*}
	&\calC_Q\bbr{\calN^1}=\calC_Q\bbr{\Psi^{\text{dph}}_{\eta_1}\otimes\Psi^{\text{dpo}}_p\otimes\Psi^{\text{dph}}_{\eta_2}}\\
	=&\calC_Q\bbr{\Psi^{\text{dph}}_{\eta_1}}+\calC_Q\bbr{\Psi^{\text{dpo}}_p}+\calC_Q\bbr{\Psi^{\text{dph}}_{\eta_2}}=1.
	\end{align*}
	By our choice of $ \eta_1 $, $ Q_C\bbr{\Phi^{\text{dph}}_{\eta_1}\otimes\Phi^{\text{dpo}}_p} $ is strictly concave in $ C $ for $ 0<C<1 $. $ Q_C\bbr{\Phi^{\text{dph}}_{\eta_2}} $ is also strictly concave in $ C $. Thus $ Q_C\bbr{N^1} $ is strictly concave in $ C $, for $ 0<C<1 $.
	
	In this case, the corresponding $ \Psi^0 $ is
	\be
	\Psi^0=\calI\otimes\bbr{\Psi^{\text{dpo}}_1}^{\otimes 2},\nonumber
	\ee
	\textit{i.e.,} a noiseless channel tensor two copies of the  complete qubit depolarizing channel. $ \calN^0 $ is a unital extension of $ \Psi^0 $.
\end{enumerate}

\subsection{Additive CE, Superadditive Q and CQE}\label{Subsec:CEadd}
Here we construct a channel that has an additive CE trade-off capacity region, but a superadditive quantum capacity, hence a superadditive quantum dynamic capacity region.

Let $ \Psi^0 $ be a classical channel and $ \Psi^1 $ be the depolarizing channel $ \Psi^{\text{dpo}}_p $. $ p $ is chosen such that $ \Psi^{\text{dpo}}_p $ has a superadditive quantum capacity. Also, we require
\be
\calC_C\bbr{\Psi^0}>C_E\bbr{\Psi^1}.\label{eq:CEineq}
\ee

Now consider the switch channel $ \calN $, consisting of $ \calN^0 $ and $ \calN^1 $, which are unital extensions of $ \Psi^0 $ and $ \Psi^1 $. It can be easily shown that unital extension does not change the classical capacity with umlimited entanglement assistance of the qubit depolarizing channel. Thus Eq. (\ref{eq:CEineq}) implies
\be
\calC_{CE}\bbr{\calN^0}\supseteq\calC_{CE}\bbr{\calN^1}\supseteq\calC_{CE}\1\bbr{\calN^1}.
\ee
Hence
\begin{align*}
\calC_{CE}\bbr{\calN}&=\rmconv \bbr{\calC_{CE}\bbr{\calN^0},\calC_{CE}\bbr{\calN^1}}=\calC_{CE}\bbr{\calN^0}\\
&=\rmconv \bbr{\calC_{CE}\bbr{\calN^0},\calC_{CE}\1\bbr{\calN^1}}=\calC_{CE}\1\bbr{\calN},
\end{align*}
\textit{i.e.,} its CE trade-off capacity region is additive.

Since $ \calC_Q\bbr{\calN^0}=0 $, it is clear that the quantum capacity of $ \calN $ is the same as that of $ \calN^1 $, which is superadditive.

Note that $ \calN $ is a unitally extended channel. This fact will be implicitly used in Section \ref{Subsec:CEnQadd}.

\subsection{Additive CE and Q, Superadditive CQE}\label{Subsec:CEnQadd}
Previously in Section \ref{Subsec:CnQaddCQnonadd}, we give an example of a channel with an additive classical and quantum capacity, but whose CQ trade-off curve is superadditive. It is unclear if the channel has an additive CE trade-off capacity region, because the CE trade-off capacity region of the depolarizing channel has not been shown to be additive. This is itself an interesting question but we'll not explore it here.

We replace $ \Psi^{\text{dpo}}_p $ in the original argument of Section \ref{Subsec:CnQaddCQnonadd} by the channel constructed in Section \ref{Subsec:CEadd}. It's clear that the rest of the argument is not changed and $ \calN $ still has a superadditive CQ trade-off capacity region.

Now both $ \calN^0 $ and $ \calN^1 $ have an additive CE trade-off capacity region. It's clear that
\begin{align*}
\calC_{CE}\bbr{\calN}&=\rmconv\bbr{\calC_{CE}\bbr{\calN^0},\calC_{CE}\bbr{\calN^1}}\\
&=\rmconv\bbr{\calC_{CE}\1\bbr{\calN^0},\calC_{CE}\1\bbr{\calN^1}}=\calC_{CE}\1\bbr{\calN},
\end{align*}
\textit{i.e.,} the CE trade-off capacity region of $ \calN $ is additive.
\subsection{Additive CQ, Superadditive CQE}\label{Subsec:CQadd}
Our construction in Section \ref{Subsec:CaddCEnonadd} has a superadditive CE trade-off capacity region. But most likely its CQ trade-off capacity region is also superadditive. This is because in Section \ref{Subsec:CaddCEnonadd}, $ \calN^0 $ is the unital extension of a classical channel, and its CQ trade-off capacity region is trivial. Hence the CQ trade-off capacity region of $ \calN $ is given by that of $ \calN^1 $, which is most likely superadditive as well.

To achieve an additive CQ trade-off capacity region, we have to substitute $ \calN^0 $ with a channel that has a non-trivial CQ trade-off capacity region.

Recall that our construction in Section \ref{Subsec:CaddCEnonadd} requires $ \calN^0 $ and $ \calN^1 $ to have properties \ref{Prop:CaddCEnonadd1}-\ref{Prop:CaddCEnonadd3}.
\begin{comment}
\begin{enumerate}
	\item $ \calN^1 $ has a superadditive CE  trade-off region, meaning
	\be
	\calC_{CE}\bbr{\calN^1}\supsetneq\calC_{CE}\1\bbr{\calN^1}.\nonumber
	\ee
	$ \calC_{CE}\bbr{\calN^1} $ is strictly concave and superadditive at a boundary point with entanglement consumption $ P $.
	\item 
	\be
	\calC_{CE}\bbr{\calN^0}\subsetneq\calC_{CE}\bbr{\calN^1}\nonumber
	\ee
	in the sense the CE trade-off region of $ \calN^0 $ is strictly smaller than that of $ \calN^1 $ when entanglement consumption is at $ P $. $ \calN^0 $ and $ \calN^1 $ have the same classical capacity.
	\item $ \calN^0 $ has an additive CE trade-off capacity region,meaning
	\be
	\calC_{CE}\bbr{\calN^0}=\calC_{CE}\1\bbr{\calN^0}.\nonumber
	\ee
\end{enumerate}
\end{comment}
These three properties ensure that $ \calN $ will have a superadditive CE trade-off capacity region, while its classical capacity is still additive.

In extending to a channel with an additive CQ trade-off capacity region, the additional properties we need are
\begin{enumerate}[label=\text{(G\arabic*)}]
	\item \label{G1}
	$\calC_{CQ}\bbr{\calN^0}\supseteq\calC_{CQ}\bbr{\calN^1}.$
	\end{enumerate}

Property \ref{G1} and \ref{Prop:General} ensure the CQ trade-off capacity region of $ \calN $ is additive, as
	\begin{align*}
	\calC_{CQ}\bbr{\calN}&=\rmconv\bbr{\calC_{CQ}\bbr{\calN^0},\calC_{CQ}\bbr{\calN^1}}=\calC_{CQ}\bbr{\calN^0}\\
	&=\calC_{CQ}\1\bbr{\calN^0}=\rmconv\bbr{\calC_{CQ}\1\bbr{\calN^0},\calC_{CQ}\1\bbr{\calN^1}}\\
	&=\calC_{CQ}\1\bbr{\calN}.
	\end{align*}

Unfortunately, we cannot find quantum channels $ \calN^0 $ and $ \calN^1 $ that satisfy all the properties. Hence we do not have an explicit construction in this case. This is because there are very few channels that we understand their dynamic capacity regions. This leaves us with a limited choice of candidates for $ \calN^0 $. However, in principle there is no obstacle and the construction will be readily available once we have a better understanding of quantum channels.
\section{Conclusion}
%The conclusion goes here.

Unlike previous studies on additivity of single resource channel capacity, our work aimed to understand how additivity of single or double resource capacity regions will effect additivity of a general resource trade-off capacity. 
In contrast to the two known results in the literature; namely, (i) additivity of the quantum capacity implies additivity of the entanglement-assisted quantum capacity region and (ii) additivity of classical-quantum and classical-entanglement capacity regions implies additivity of the three resource capacity region, the additivity of all the remaining situations does not hold. In this work, we  identified  all possible occurrences where superadditivity could occur in the trade-off quantum dynamic capacity. Furthermore, we provided an explicit construction of quantum channels for most instances. 
Our main technical tool combines properties of switch channels and unital extension of known quantum channels. 

An obvious open question is an explicit construction of a quantum channel whose classical-quantum capacity region is additive, but its triple trade-off capacity is superadditive. Moreover, there are other triple resource trade-off capacity regions \cite{HsiehWilde10, Wilde12}. Could similar statements made in this work hold in these scenarios as well?

\section*{Acknowledgment}
EYZ and PWS are supported by the National Science Foundation under grant Contract Number CCF-1525130. QZ is supported by the Claude E. Shannon Research Assistantship. MH is supported by an ARC Future Fellowship under Grant FT140100574. PWS is supported by the NSF through the STC for Science of Information under grant number CCF0-939370.

\appendices
\section{Proof of Lemma~\ref{lem:unitalHadamard}}
\begin{proof}
Consider $ \Phi^0_{RC\to B^0} $ and $ \Psi^1_{A^1\to B^1} $, where $ \Phi^0 $ is a unital extension of a Hadamard channel $ \Psi^0_{C\to B^0} $, and $ \Psi^1 $ is an arbitrary channel.

The result follows if both the CQ and CE trade-off capacity regions of $\Phi^0$ are additive \cite{HsiehWilde10}. To show that the CQ trade-off capacity region is additive for $\Phi^0$, it was shown in Ref.~\cite{Bradler10} it suffices to prove that 
\begin{equation}
  f_\lambda\bbr{\Phi^0\otimes\Psi^1} = f_\lambda\bbr{\Phi^0}+ f_\lambda\bbr{\Psi^1}
\end{equation}
 for any channel $ \Psi^1 $, where
\begin{equation}
 f_\lambda\bbr{\calN} = \max_{\rho} I(X;B)_\sigma + \lambda I(A\rangle BX)_\sigma.
\end{equation}
The state $\sigma$ is the channel output state with $\rho$ being the input state (see, e.g., Theorem~\ref{thm_CQE}). In the following, we will only show that $f_\lambda\bbr{\Phi^0\otimes\Psi^1} \leq f_\lambda\bbr{\Phi^0}+ f_\lambda\bbr{\Psi^1}$ because the other direction is trivial from its definition.

Since $\Phi^0\otimes\Psi^1:CRA^1\to B^0B^1$ is a partial cq channel, then by the same argument as that in Lemma~\ref{lem:cqCQE}, $f_\lambda\bbr{\Phi^0\otimes\Psi^1}$ can be achieved with input states of the following form
\be
\rho_{XRACA^1}=\sum_{x,j}\frac{p(x)}{|R|}\ket{x,j}\bra{x,j}_X\otimes\ket{j}\bra{j}_R\otimes\phi_{ACA^1}^x,\nonumber
\ee
with output states
\be
\sigma_{XAB^0B^1}=\sum_{x,j}\frac{p(x)}{|R|}\ket{x,j}\bra{x,j}_X\otimes\sigma^{xj}_{AB^0B^1},  \label{eq:Hadamart_sigma}
\ee
where
\be
\sigma^{xj}_{AB^0B^1}=\Phi^0\otimes\Psi^1\bbr{\ket{j}\bra{j}_R\otimes\phi_{ACA^1}^x}.\nonumber
\ee

Let $ U^0_{C\to B^0E^0} $ and $ U^1_{A^1\to B^1E^1} $ be the isometric extensions of $ \Psi^0 $ and $ \Psi^1 $, and let
\begin{align*}
\varrho_{XACA^1}&=\sum_{x}p(x)\ket{x}\bra{x}_X\otimes\phi_{ACA^1}^x    \\
\omega_{XAA^1B^0E^0}&=\bbr{U^0\otimes I}\varrho_{XACA^1}\bbr{U^0\otimes I}^\dagger \\    
\varsigma_{XAB^0B^1E^0E^1}&=\bbr{U^0\otimes U^1}\varrho_{XACA^1}\bbr{U^0\otimes U^1}^\dagger. 
\end{align*}
Moreover, let
\be
\theta_{XYAB^1E^0E^1}=\calD^1_{B^0\to Y}\bbr{\varsigma_{XAB^0B^1E^0E^1}},\nonumber
\ee
where $\calD^2_{Y\to E^0}\circ \calD^1_{B^0\to Y}=\calD_{B^0\to E^0} $ is a degrading map for the Hadamard channel $ \Psi^0 $.

For any state $ \sigma_{XAB^0B^1} $ in Eq.~(\ref{eq:Hadamart_sigma}), we have
\begin{align*}
&f_\lambda\bbr{\Phi^0\otimes\Psi^1}\\
=&I\bbr{X;B^0B^1}_\sigma+\lambda I\bbr{A\rangle B^0B^1X}_\sigma\\
%=&S\bbr{B^0B^1}_\sigma+\sum_x\frac{p(x)}{|R|}\left[(\lambda-1)S\bbr{B^0B^1}_{\sigma^{xj}}-\lambda S\bbr{AB^0B^1}_{\sigma^{xj}}\right]\\
=&S\bbr{B^0B^1}_\sigma+\left[(\lambda-1)S\bbr{B^0B^1|X}_{\sigma}-\lambda S\bbr{AB^0B^1|X}_{\sigma}\right]\\
=&S\bbr{B^0B^1}_\varsigma+\left[(\lambda-1)S\bbr{B^0B^1|X}_{\varsigma}-\lambda S\bbr{AB^0B^1|X}_{\varsigma}\right], 
\end{align*}
where the last equality follows from the same argument used in Eqs.~(\ref{eq:lem3_1}) and (\ref{eq:lem3_3}). Then subadditivity of the von Neumann entropy and chain rule yield
\begin{align*}
\leq &S\bbr{B^0}_\varsigma+(\lambda-1)S\bbr{B^0|X}_{\varsigma}-\lambda S\bbr{E^0|X}_{\varsigma}\\
+&S\bbr{B^1}_\varsigma+(\lambda-1)S\bbr{B^1|B^0X}_{\varsigma}-\lambda S\bbr{E^1|E^0X}_{\varsigma}\\ 
\leq &S\bbr{B^0}_\varsigma+(\lambda-1)S\bbr{B^0|X}_{\varsigma}-\lambda S\bbr{E^0|X}_{\varsigma}\\
+&S\bbr{B^1}_\theta+(\lambda-1)S\bbr{B^1|XY}_{\theta}-\lambda S\bbr{E^1|XY}_{\theta}
\end{align*}
where the last inequality uses the fact that $ S\bbr{B^1|B^0X}_{\varsigma}\leq S\bbr{B^1|YX}_{\theta}$ due to the existence of $ \calD^1 $ and $ S\bbr{E^1|E^0X}_{\varsigma}
\geq S\bbr{E^1|YX}_{\theta} $ due to the existence of $ \calD^2 $. Finally,
\begin{align*}
=&\bbr{I\bbr{X;B^0}_{\omega}+\lambda I\bbr{AA^1\rangle B^0X}_{\omega}}\\
+&\bbr{I\bbr{XY;B^1}_{\theta}+\lambda I\bbr{AE^0\rangle B^1XY}_{\theta}}\\
\leq& f_\lambda\bbr{\Phi^0}+f_\lambda\bbr{\Psi^1}
\end{align*}
because $S(E^0|X)_\varsigma = S(AA^1B^0|X)_\omega$ and $S\bbr{E^1|XY}_{\theta} = S\bbr{AB^1E^0|XY}_{\theta}$.

To prove that the CE trade-off capacity region of the channel $\Phi^0$ is additive is equivalent to showing that \cite{Bradler10}:
\be
g_\lambda\bbr{\Phi^0\otimes \Psi^1} = g_\lambda\bbr{\Phi^0}+ g_\lambda\bbr{\Psi^1},
\ee
where $0\leq \lambda <1$, 
\begin{equation}
g_\lambda\bbr{\calN}=\max_\sigma I(AX;B)_\sigma - \lambda S(A|X)_\sigma
\end{equation}
and $\sigma$ is of the form given in Eq.~(\ref{eq:cqinput}).

However this proof proceeds similarly; hence, we will omit it.
\end{proof}

\section{Proof of Lemma \ref{lem:unitalclassical}}
\textit{Lemma 6}:
If $ \Psi^0 $ is a classical channel, then the dynamic capacity region is additive for $ \Psi^0\otimes\Psi^1 $, for arbitrary $ \Psi^1 $. 
Moreover, the dynamic capacity region of $ \Psi^0 $ is described by the following relation
\begin{align*}
C+2Q&\leq \calC_C\bbr{\Psi^0},\\
Q+E&\leq 0,\\
C+Q+E&\leq \calC_C\bbr{\Psi^0},
\end{align*}
where $ \calC_C\bbr{\Psi^0} $ is the classical capacity of $ \Psi^0 $.

The same holds for a unital extension of a classical channel.

\begin{proof}
Consider the 1-shot dynamic capacity region of $ \Psi^0_{A^{0'}\to B^0} $. By Lemma \ref{lem:cqCQE}, $ \calC_{CQE}\1\bbr{\Psi^0} $ can be achieved with respect to cq states $ \sigma_{XA^0B^0}=\Psi^0\bbr{\rho_{XA^0A^{0'}}} $, where $ \rho_{XA^0A^{0'}} $ is of the form
\be
\rho_{XA^0A^{0'}}=\sum_{x,j}p(x,j)\ket{x,j}\bra{x,j}_X\otimes\ket{j}\bra{j}_{A^{0'}}\otimes\phi_{A^0}^{xj}.\nonumber
\ee
Thus
\be
\sigma_{XA^0B^0}=\sum_{x,j}p(x,j)\ket{x,j}\bra{x,j}_X\otimes\sigma^{xj}_{A^0B^0},\nonumber
\ee
where
\be
\sigma^{xj}_{A^0B^0}=\Psi^0_{A^{0'}\to B^0}\bbr{\ket{j}\bra{j}_{A^{0'}}\otimes\phi^{xj}_{A^0}}\nonumber
\ee
is now a product state with respect to $ A^0 $ and $ B^0 $.

The three entropic quantities of interest can be simplied when evaluated with respect to $ \sigma_{XA^0B^0} $, as
\begin{align*}
I\bbr{A^0X;B^0}_\sigma&=S\bbr{B^0}_\sigma-\sum_{x,j}p(x,j)S\bbr{B^0}_{\sigma^{xj}}\leq \calC_C\bbr{\Psi^0},\\
I\bbr{A^0\rangle B^0X}_\sigma&=-\sum_{x,j}p(x,j)S\bbr{A^0}_{\sigma^{xj}}\leq 0,\\
I\bbr{X;B^0}_\sigma&\leq \calC_C\bbr{\Psi^0}.
\end{align*}
It's also clear that those inequalities can be achieved. Thus $ \calC_{CQE}\1\bbr{\Psi^0} $ is described by
\begin{align*}
C+2Q&\leq \calC_C\bbr{\Psi^0},\\
Q+E&\leq 0,\\
C+Q+E&\leq \calC_C\bbr{\Psi^0}.
\end{align*}
Since the classical capacity of a classical channel is additive, the dynamic capacity region of $ \Psi^0 $ is additive and is described by the same set of inequalities.

Next we show that the dynamic capacity region is additive for $ \Psi^0 $ and $ \Psi^1 $, with $ \Psi^1 $ arbitrary.

Since $ \Psi^0_{A^{0'}\to B^0}\otimes\Psi^1_{A^{1'}\to B^1} $ is a partial cq channel, its 1-shot dynamic capacity region $ \calC_{CQE}\1\bbr{\Psi^0\otimes\Psi^1} $ can be achieved with respect to cq states $ \sigma_{XAB^0B^1}=\Psi^0_{A^{0'}\to B^0}\otimes\Psi^1_{A^{1'}\to B^1}\bbr{\rho_{XAA^{0'}A^{1'}}} $, where $ \rho_{XAA^{0'}A^{1'}} $ is of the form
\be
\rho_{XAA^{0'}A^{1'}}=\sum_{x,j}p(x,j)\ket{x,j}\bra{x,j}_X\otimes\ket{j}\bra{j}_{A^{0'}}\otimes\phi^{xj}_{AA^{1'}}.\nonumber
\ee
For $ \rho_{XAA^{0'}A^{1'}} $ of this form, $ \sigma_{XAB^0B^1} $ is of the form
\be
\sigma_{XAB^0B^1}=\sum_{x,j}p(x,j)\ket{x,j}\bra{x,j}_X\otimes\sigma^{xj}_{AB^0B^1},\nonumber
\ee
with
\be
\sigma^{xj}_{AB^0B^1}=\Psi^0_{A^{0'}\to B^0}\otimes\Psi^1_{A^{1'}\to B^1}\bbr{\ket{j}\bra{j}_{A^{0'}}\otimes\phi^{xj}_{AA^{1'}}}.\nonumber
\ee
For such $ \sigma_{XAB^0B^1} $, each of the three entropic quantities have simple upper bounds,
\begin{align*}
I\bbr{AX;B^0B^1}_\sigma&\leq I\bbr{X;B^0}_\sigma+I\bbr{AX;B^1}_\sigma,\\
I\bbr{A\rangle B^0B^1X}_\sigma&=I\bbr{A\rangle B^1X}_\sigma,\\
I\bbr{X;B^0B^1}_\sigma&\leq I\bbr{X;B^0}_\sigma+I\bbr{X;B^1}_\sigma,
\end{align*}
where we've used subadditivity of the von Neumann entropy.
Thus the 1-shot dynamic capacity region of $ \Psi^0\otimes\Psi^1 $ has a simple upper bound
\be
\calC_{CQE}\1\bbr{\Psi^0\otimes\Psi^1}\subseteq \calC_{CQE}\1\bbr{\Psi^0}+\calC_{CQE}\1\bbr{\Psi^1}.\nonumber
\ee
It's trivial to extend it to the dynamic capacity region of $ \Psi^0\otimes\Psi^1 $
\be
\calC_{CQE}\bbr{\Psi^0\otimes\Psi^1}\subseteq \calC_{CQE}\bbr{\Psi^0}+\calC_{CQE}\bbr{\Psi^1}.\nonumber
\ee
Since the other direction of inclusion is obvious, we have
\be
\calC_{CQE}\bbr{\Psi^0\otimes\Psi^1}= \calC_{CQE}\bbr{\Psi^0}+\calC_{CQE}\bbr{\Psi^1}.\nonumber
\ee

For unital extensions of a classical channel, we observe that, if the Heisenberg-Weyl operators are defined on the standard basis for the output of the channel, then the resulting channel is also a classical channel. Hence the above result applies.
\end{proof}
\section{unital extension of the qubit dephasing channel and $ 1\to N $ cloning channel}\label{Appendix:dephasing}
\begin{lemma}\label{lem:dephasing}
The CE and CQ trade-off curve of the qubit dephasing channel and $ 1\to N $ cloning channels are unchanged after a unital extension.
\end{lemma}
\begin{proof}
Consider the qubit dephasing channel $ \Psi^{\textrm{dph}}_\eta $ and a $ 1\to N $ cloning channel $ \Psi^{1\to N} $, and their unital extensions $ \Phi^{\textrm{dph}}_\eta $ and $ \Phi^{1\to N} $. The statement of this lemma is equivalent to showing that
\begin{align*}
&f_\lambda\bbr{\Psi}=f_\lambda\bbr{\Phi}~~~\forall \lambda\geq 1,\\
&g_\lambda\bbr{\Psi}=g_\lambda\bbr{\Phi}~~~\forall 0\leq\lambda< 1.
\end{align*}
for
\be
(\Psi,\Phi)=\bbr{\Psi^{\textrm{dph}}_\eta,\Phi^{\textrm{dph}}_\eta},\bbr{\Psi^{1\to N},\Phi^{1\to N}}.\nonumber
\ee

In Lemma \ref{lem:unitaldynamic}, we have argued that the 1-shot dynamic capacity region of a unitally extended channel can be achieved with input of the form in Eq. (\ref{eq:lem3_cqinput}).
Evaluating $ f_\lambda\bbr{\Phi} $ on such states, one obtains
\begin{align}
f_\lambda\bbr{\Phi}&=\log(|B|)+(\lambda-1)S\bbr{B|X'}_{\sigma}-\lambda S\bbr{AB|X'}_{\sigma}\nonumber\\
&=\log(|B|)+\sum_{x,j,k}p(x,j,k)\left[(\lambda-1)S(B)_{\sigma^{xjk}}-\lambda S(AB)_{\sigma^{xjk}}\right]\nonumber\\
&=\log(|B|)+\sum_{x,j}p(x,j)\left[(\lambda-1)S(B)_{\sigma^{xj}}-\lambda S(AB)_{\sigma^{xj}}\right]\nonumber\\
&\leq \log(|B|)+\max_\sigma\left[(\lambda-1)S\bbr{B}_\sigma-\lambda S\bbr{AB}_\sigma\right],\label{eq:dephasingCQoutput}
\end{align}
where
\be\label{eq:lem7output}
\sigma_{AB}=\Psi_{C\to B}\bbr{\phi_{AC}}.
\ee

For such a $ \sigma_{AB}=\Psi\bbr{\phi_{AC}} $ that achieves Eq. (\ref{eq:dephasingCQoutput}), one can construct
\be
\rho_{XAA'}=\frac{1}{|R|}\ket{k}\bra{k}_X\otimes\ket{k}\bra{k}_R\otimes\phi_{AC}.
\ee
This state will saturate the above inequality.
\begin{comment}
Consider $ \rho_{A'}={\rm tr}\bbr{\phi_{AA'}} $. Then $ S\bbr{AB}_\sigma=S\bbr{\bbr{\Psi^{\textrm{dph}}_\eta}^c\bbr{\rho_{A'}}} $.
$ \rho_{A'} $ has the Bloch representation
\be
\rho_{A'}=\frac{I+\vec{r}\cdot\vec{\sigma}}{2}.\nonumber
\ee
So
\be
\Psi^{\textrm{dph}}_\eta\bbr{\rho_{A'}}=\frac{I+\vec{r}'\cdot\vec{\sigma}}{2}\nonumber
\ee
where
\be
\vec{r}'=\left((1-2\eta)r_1,(1-2\eta)r_2,r_3\right).\nonumber
\ee
So
\be
S(B)_\sigma=H_2\bbr{\frac{1+r'}{2}}\nonumber
\ee
and
\begin{align*}
&S\bbr{\bbr{\Psi^{\textrm{dph}}_\eta}^c\bbr{\rho_{A'}}}\\
=&H_2\bbr{\frac{1}{2}\bbr{1+\sqrt{1-4\eta(1-\eta)(1+r_3)(1-r_3)}}}.
\end{align*}
So it's obvious that for $ (\lambda-1)S\bbr{B}_\sigma-\lambda S\bbr{AB}_\sigma $, the first term can be maximized by setting $ r_1=r_2=0 $, without changing the second term. Thus diagonal states maximize the function.
It is obvious then that this generates the same CQ trade-off curve as $ f_\lambda\bbr{\Psi^{\textrm{dph}}_\eta} $.
\end{comment}

For $ \Psi^{\textrm{dph}}_\eta $ and $ \Psi^{1\to N} $, it can be verified \cite{Bradler10} that their $ f_\lambda $ have the same form, \textit{i.e.,}
\be
f_\lambda\bbr{\Psi}=\log(|B|)+\max_\sigma\left[(\lambda-1)S\bbr{B}_\sigma-\lambda S\bbr{AB}_\sigma\right],
\ee
with $ \sigma $ of the form given in Eq. (\ref{eq:lem7output}).

The same argument also applies to $ g_\lambda $.
\end{proof}

The CQ trade-off curve of the qubit dephasing channel was computed in Ref. \cite{DevetakShor05}, and the CE trade-off curve was computed in Ref. \cite{Hsieh10}. The CE and CQ trade-off curves of the $ 1\to N $ cloning channel were given in Ref. \cite{Bradler10}. Other than the special cases ($ \eta=0, 1/2 $ for the dephasing channel, $ N=1 $ for the $ 1\to N $ cloning channel), it can be verified that their CE and CQ trade-off curves are strictly concave at every point. By Lemma \ref{lem:dephasing}, this property is true for their unital extensions.

\section{Convex hull}\label{Appendix:convexhullclosure}
Here we show that\footnote{$ \calN^0 $ and $ \calN^1 $ are assumed to be finite-dimensional.}
\begin{align*}
&\overline{\rmconv\bbr{\calC_{CQE}\bbr{\calN^0},\calC_{CQE}\bbr{\calN^1}}}\\
=&\rmconv\bbr{\calC_{CQE}\bbr{\calN^0},\calC_{CQE}\bbr{\calN^1}}.  
\end{align*}
We quote a few properties about convex hull and Minkowski addition that we will use \cite{Schneider93}: (i) For two closed sets $A$ and $B$ in $ \mathbb{R}^k $, if $A$ is bounded, then $ A+B $ is closed. (ii)  For two sets $ A $ and $ B $ in $ \mathbb{R}^k $, $ \rmconv(A+B)=\rmconv(A)+\rmconv(B) $ .  (iii) The convex hull of a bounded set in $ \mathbb{R}^k $ is also bounded. 

First, we note that, by Ref. \cite{HsiehWilde10}, all points in the 1-shot dynamic capacity region can be achieved by the classically enhanced father protocol, combined with unit protocols, \textit{i.e.,}
\be
\calC_{CQE}\1\bbr{\calN}=\bigcup_{\sigma}\calC_{CQE,\text{CEF}}\1\bbr{\calN}_\sigma+\calC_{CQE,\text{unit}},\nonumber
\ee
where 
\be
\calC\1_{CQE,\text{CEF}}\bbr{\calN}_\sigma=\{I(X;B)_\sigma,\frac{1}{2}I(A;B|X)_\sigma,-\frac{1}{2}I(A;E|X)_\sigma\}\nonumber
\ee
is the rate achieved using the classically enhanced father protocol, and $ \sigma $ is of the form in Eq. (\ref{eq:cqinput}). $ \calC_{CQE,\text{unit}} $ are all the rates achieved by the unit protocols. Clearly $ \calC_{CQE,\text{unit}} $ is convex and closed.
Define
\be
\calC\1_{CQE,\text{CEF}}\bbr{\calN}=\bigcup_{\sigma}\calC_{CQE,\text{CEF}}\1\bbr{\calN}_\sigma.\nonumber
\ee
Clearly $ \calC\1_{CQE,\text{CEF}}\bbr{\calN} $ is bounded by the input and output dimensions of $ \calN $.

Then
\begin{align*}
\calC_{CQE}\bbr{\calN}&=\overline{\bigcup_{k=1}\frac{1}{k}\bbr{\calC\1_{CQE,\text{CEF}}\bbr{\calN^{\otimes k}}+\calC_{CQE,\text{unit}}}}\\
&=\overline{\bigcup_{k=1}\frac{1}{k}\calC\1_{CQE,\text{CEF}}\bbr{\calN^{\otimes k}}+\calC_{CQE,\text{unit}}}.
\end{align*}
Denote
\begin{align*}
A&=\bigcup_{k=1}\frac{1}{k}\calC\1_{CQE,\text{CEF}}\bbr{\calN^{\otimes k}},\\
B&=\calC_{CQE,\text{unit}}.
\end{align*}
Since $ A $ is bounded, $ B $ is closed, by (i) and (iii), $ \overline{A}+B $ is also closed.\\
Since $ A+B\subseteq \overline{A}+B $, and $ \overline{A}+B $ is closed, we have
\be
\overline{A+B}\subseteq \overline{A}+B.\nonumber
\ee
It is also obvious that
\be
\overline{A+B}\supseteq \overline{A}+B,\nonumber
\ee
hence
\be
\overline{A+B}=\overline{A}+B.\nonumber
\ee
Denote
\be
\calC_{CQE,\text{CEF}}\bbr{\calN}=\overline{\bigcup_{k=1}\frac{1}{k}\calC\1_{CQE,\text{CEF}}\bbr{\calN^{\otimes k}}}.\nonumber
\ee
Then by the above arguments,
\be
\calC_{CQE}\bbr{\calN}=\calC_{CQE,\text{CEF}}\bbr{\calN}+\calC_{CQE,\text{unit}}.\nonumber
\ee
Now we apply the above result to $ \calN^0 $ and $ \calN^1 $.
\begin{align*}
&\rmconv\bbr{\calC_{CQE}\bbr{\calN^0},\calC_{CQE}\bbr{\calN^1}}\\
=&\rmconv\bbr{\calC_{CQE}\bbr{\calN^0}\cup\calC_{CQE}\bbr{\calN^1}}\\
=&\rmconv\bbr{\calC_{CQE,\text{CEF}}\bbr{\calN^0}\cup\calC_{CQE,\text{CEF}}\bbr{\calN^1}+\calC_{CQE,\text{unit}}}\\
=&\rmconv\bbr{\calC_{CQE,\text{CEF}}\bbr{\calN^0}\cup\calC_{CQE,\text{CEF}}\bbr{\calN^1}}+\calC_{CQE,\text{unit}}.
\end{align*}
In the last line, we used (ii).

Since $ \calC_{CQE,\text{CEF}}\bbr{\calN} $ is closed and bounded for any finite dimensional quantum channel $ \calN $, the same must be true for $ \calC_{CQE,\text{CEF}}\bbr{\calN^0}\cup\calC_{CQE,\text{CEF}}\bbr{\calN^1} $. Hence $ \rmconv\bbr{\calC_{CQE,\text{CEF}}\bbr{\calN^0}\cup\calC_{CQE,\text{CEF}}\bbr{\calN^1}} $ is closed and bounded. Thus $ \rmconv\bbr{\calC_{CQE,\text{CEF}}\bbr{\calN^0}\cup\calC_{CQE,\text{CEF}}\bbr{\calN^1}}+\calC_{CQE,\text{unit}} $ is also closed.
% Generated by IEEEtran.bst, version: 1.14 (2015/08/26)


\begin{thebibliography}{10}
	\providecommand{\url}[1]{#1}
	\csname url@samestyle\endcsname
	\providecommand{\newblock}{\relax}
	\providecommand{\bibinfo}[2]{#2}
	\providecommand{\BIBentrySTDinterwordspacing}{\spaceskip=0pt\relax}
	\providecommand{\BIBentryALTinterwordstretchfactor}{4}
	\providecommand{\BIBentryALTinterwordspacing}{\spaceskip=\fontdimen2\font plus
		\BIBentryALTinterwordstretchfactor\fontdimen3\font minus
		\fontdimen4\font\relax}
	\providecommand{\BIBforeignlanguage}[2]{{%
			\expandafter\ifx\csname l@#1\endcsname\relax
			\typeout{** WARNING: IEEEtran.bst: No hyphenation pattern has been}%
			\typeout{** loaded for the language `#1'. Using the pattern for}%
			\typeout{** the default language instead.}%
			\else
			\language=\csname l@#1\endcsname
			\fi
			#2}}
	\providecommand{\BIBdecl}{\relax}
	\BIBdecl
	
	\bibitem{Shannon48}
	C.~E. Shannon, ``A mathematical theory of communication,'' \emph{The Bell
		System Technical Journal}, vol.~27, no.~3, pp. 379--423, July 1948.
	
	\bibitem{Devetak04}
	\BIBentryALTinterwordspacing
	I.~Devetak, A.~W. Harrow, and A.~Winter, ``A family of quantum protocols,''
	\emph{Phys. Rev. Lett.}, vol.~93, p. 230504, Dec 2004.
	\BIBentrySTDinterwordspacing
	
	\bibitem{Devetak08}
	I.~Devetak, A.~W. Harrow, and A.~J. Winter, ``A resource framework for quantum
	shannon theory,'' \emph{IEEE Transactions on Information Theory}, vol.~54,
	no.~10, pp. 4587--4618, Oct 2008.
	
	\bibitem{HsiehWilde10}
	M.-H. Hsieh and M.~M. Wilde, ``Trading classical communication, quantum
	communication, and entanglement in quantum shannon theory,'' \emph{IEEE
		Transactions on Information Theory}, vol.~56, no.~9, pp. 4705--4730, Sept
	2010.
	
	\bibitem{Bennett92}
	\BIBentryALTinterwordspacing
	C.~H. Bennett and S.~J. Wiesner, ``Communication via one- and two-particle
	operators on einstein-podolsky-rosen states,'' \emph{Phys. Rev. Lett.},
	vol.~69, pp. 2881--2884, Nov 1992.
	\BIBentrySTDinterwordspacing
	
	\bibitem{Bennett99}
	\BIBentryALTinterwordspacing
	C.~H. Bennett, P.~W. Shor, J.~A. Smolin, and A.~V. Thapliyal,
	``Entanglement-assisted classical capacity of noisy quantum channels,''
	\emph{Phys. Rev. Lett.}, vol.~83, pp. 3081--3084, Oct 1999.
	\BIBentrySTDinterwordspacing
	
	\bibitem{Bennett02}
	C.~H. Bennett, P.~W. Shor, J.~A. Smolin, and A.~V. Thapliyal,
	``Entanglement-assisted capacity of a quantum channel and the reverse shannon
	theorem,'' \emph{IEEE Transactions on Information Theory}, vol.~48, no.~10,
	pp. 2637--2655, Oct 2002.
	
	\bibitem{Hsieh10}
	M.-H. Hsieh and M.~M. Wilde, ``Entanglement-assisted communication of classical
	and quantum information,'' \emph{IEEE Transactions on Information Theory},
	vol.~56, no.~9, pp. 4682--4704, Sept 2010.
	
	\bibitem{Datta13}
	N.~Datta and M.~H. Hsieh, ``One-shot entanglement-assisted quantum and
	classical communication,'' \emph{IEEE Transactions on Information Theory},
	vol.~59, no.~3, pp. 1929--1939, March 2013.
	
	\bibitem{Holevo98}
	A.~S. Holevo, ``The capacity of the quantum channel with general signal
	states,'' \emph{IEEE Transactions on Information Theory}, vol.~44, no.~1, pp.
	269--273, Jan 1998.
	
	\bibitem{Schumacher97}
	\BIBentryALTinterwordspacing
	B.~Schumacher and M.~D. Westmoreland, ``Sending classical information via noisy
	quantum channels,'' \emph{Phys. Rev. A}, vol.~56, pp. 131--138, Jul 1997.
	\BIBentrySTDinterwordspacing
	
	\bibitem{Lloyd97}
	\BIBentryALTinterwordspacing
	S.~Lloyd, ``Capacity of the noisy quantum channel,'' \emph{Phys. Rev. A},
	vol.~55, pp. 1613--1622, Mar 1997.
	\BIBentrySTDinterwordspacing
	
	\bibitem{Shor02}
	P.~W. Shor, ``The quantum channel capacity and coherent information,''
	\emph{MSRI Workshop on Quantum Computation}, 2002.
	
	\bibitem{Devetak05}
	I.~Devetak, ``The private classical capacity and quantum capacity of a quantum
	channel,'' \emph{IEEE Transactions on Information Theory}, vol.~51, no.~1,
	pp. 44--55, Jan 2005.
	
	\bibitem{SmithYard08}
	\BIBentryALTinterwordspacing
	G.~Smith and J.~Yard, ``Quantum communication with zero-capacity channels,''
	\emph{Science}, vol. 321, no. 5897, pp. 1812--1815, 2008.
	\BIBentrySTDinterwordspacing
	
	\bibitem{Shor04}
	\BIBentryALTinterwordspacing
	P.~W. Shor, ``The classical capacity achievable by a quantum channel assisted
	by a limited entanglement.'' \emph{Quantum Information and Computation},
	vol.~4, no.~6, pp. 537 -- 545, 2004.
	\BIBentrySTDinterwordspacing
	
	\bibitem{Smith08}
	G.~Smith, J.~A. Smolin, and A.~Winter, ``The quantum capacity with symmetric
	side channels,'' \emph{IEEE Transactions on Information Theory}, vol.~54,
	no.~9, pp. 4208--4217, Sept 2008.
	
	\bibitem{Hastings09}
	\BIBentryALTinterwordspacing
	M.~B. Hastings, ``Superadditivity of communication capacity using entangled
	inputs,'' \emph{Nat Phys}, vol.~5, no.~4, pp. 255--257, 04 2009.
	\BIBentrySTDinterwordspacing
	
	\bibitem{Hayden08}
	\BIBentryALTinterwordspacing
	P.~Hayden and A.~Winter, ``Counterexamples to the maximal p-norm
	multiplicativity conjecture for all p > 1,'' \emph{Communications in
		Mathematical Physics}, vol. 284, no.~1, pp. 263--280, 2008.
	\BIBentrySTDinterwordspacing
	
	\bibitem{Shor04additivity}
	\BIBentryALTinterwordspacing
	P.~W. Shor, ``Equivalence of additivity questions in quantum information
	theory,'' \emph{Communications in Mathematical Physics}, vol. 246, no.~3, pp.
	453--472, 2004.
	\BIBentrySTDinterwordspacing
	
	\bibitem{ZhuZhuangShor17}
	\BIBentryALTinterwordspacing
	E.~Y. Zhu, Q.~Zhuang, and P.~W. Shor, ``Superadditivity of the classical
	capacity with limited entanglement assistance,'' \emph{Phys. Rev. Lett.},
	vol. 119, p. 040503, Jul 2017.
	\BIBentrySTDinterwordspacing
	
	\bibitem{DevetakShor05}
	\BIBentryALTinterwordspacing
	I.~Devetak and P.~W. Shor, ``The capacity of a quantum channel for simultaneous
	transmission of classical and quantum information,'' \emph{Communications in
		Mathematical Physics}, vol. 256, no.~2, pp. 287--303, 2005.
	\BIBentrySTDinterwordspacing
	
	\bibitem{Bradler10}
	\BIBentryALTinterwordspacing
	K.~Br\'adler, P.~Hayden, D.~Touchette, and M.~M. Wilde, ``Trade-off capacities
	of the quantum hadamard channels,'' \emph{Phys. Rev. A}, vol.~81, p. 062312,
	Jun 2010.
	\BIBentrySTDinterwordspacing
	
	\bibitem{King02}
	\BIBentryALTinterwordspacing
	C.~King, ``Additivity for unital qubit channels,'' \emph{Journal of
		Mathematical Physics}, vol.~43, no.~10, pp. 4641--4653, 2002.
	\BIBentrySTDinterwordspacing
	
	\bibitem{Divincenzo98}
	\BIBentryALTinterwordspacing
	D.~P. DiVincenzo, P.~W. Shor, and J.~A. Smolin, ``Quantum-channel capacity of
	very noisy channels,'' \emph{Phys. Rev. A}, vol.~57, pp. 830--839, Feb 1998.
	\BIBentrySTDinterwordspacing
	
	\bibitem{Fukuda07}
	\BIBentryALTinterwordspacing
	M.~Fukuda, ``Simplification of additivity conjecture in quantum information
	theory,'' \emph{Quantum Information Processing}, vol.~6, no.~3, pp. 179--186,
	2007.
	\BIBentrySTDinterwordspacing
	
	\bibitem{Elkouss15}
	\BIBentryALTinterwordspacing
	D.~Elkouss and S.~Strelchuk, ``Superadditivity of private information for any
	number of uses of the channel,'' \emph{Phys. Rev. Lett.}, vol. 115, p.
	040501, Jul 2015.
	\BIBentrySTDinterwordspacing
	
	\bibitem{Wilde13}
	M.~M. Wilde, \emph{Quantum Information Theory}.\hskip 1em plus 0.5em minus
	0.4em\relax Cambridge University Press, 2013.
	
	\bibitem{Schneider93}
	R.~Schneider, \emph{Convex Bodies: The Brunn-Minkowski Theory (Encyclopedia of
		Mathematics and its Applications)}.\hskip 1em plus 0.5em minus 0.4em\relax
	Cambridge University Press, 1993.
	
	\bibitem{Rudin76}
	\BIBentryALTinterwordspacing
	W.~Rudin, \emph{Principles of Mathematical Analysis}.\hskip 1em plus 0.5em
	minus 0.4em\relax McGraw-Hill Education, 1976.
	\BIBentrySTDinterwordspacing
	
	\bibitem{Wilde12}
	\BIBentryALTinterwordspacing
	M.~M. Wilde and M.-H. Hsieh, ``The quantum dynamic capacity formula of a
	quantum channel,'' \emph{Quantum Information Processing}, vol.~11, no.~6, pp.
	1431--1463, 2012.
	\BIBentrySTDinterwordspacing
	
\end{thebibliography}
\end{document}